\documentclass[journal,onecolumn]{IEEEtran}

\usepackage{eso-pic, graphics, graphicx, xspace, epic, eepic,theorem}

\let\labelindent\relax
\usepackage{enumitem}

\usepackage[cmex10]{amsmath}
\usepackage{amssymb,amsmath, amsfonts}
\usepackage{accents} % to put a underline bar in the definition of the vect 
\usepackage[hyphens]{url}
\usepackage{comment}
\usepackage{xcolor}
\usepackage{cite}
\newtheorem{theorem}{Theorem}
\newtheorem{lemma}{Lemma}
\newtheorem{proposition}{Proposition}

\newtheorem{remark}{Remark}

\DeclareMathOperator{\sinc}{sinc}
\usepackage{mathtools}

\newcommand{\nn}{\nonumber}
\newcommand{\integers}{\mathbb{Z}}
\newcommand{\vect}[1]{\mathbf{#1}}
\newcommand{\const}[1]{{\mathcal{#1}}}
\newcommand{\Reals}{\mathbb{R}}
\newcommand{\Complex}{\mathbb{C}}
\newcommand{\Naturals}{\mathbb N}

\newcommand{\snr}{\text{SNR}}

\newcommand{\ie}{\emph{i.e.}}

\newcommand{\der}{\mathrm{d}}

\newcommand{\E}{\mathsf{E}}

\newcommand{\U}{\mathcal U}

\newcommand{\defeq}{\overset{\text{def}}{=}}

\newcommand*{\qed}{\hfill\ensuremath{\square}}%

\newcommand{\iid}{\textnormal{i.i.d.}}

\newcommand*\xbar[1]{%
  \hbox{%
    \vbox{%
      \hrule height 0.7pt % The actual bar
      \kern0.3ex%         % Distance between bar and symbol
      \hbox{%
        \kern-0.2em%      % Shortening on the left side
        \ensuremath{#1}%
        \kern-0.0em%      % Shortening on the right side
      }%
    }%
  }%
}

\begin{document}

\title{Capacity-Achieving Input Distribution in Per-Sample Zero-Dispersion Model of Optical Fiber%
\thanks{
Jihad Fahs is with the School of Engineering, Lebanese American University, Beirut 1102 2801, Lebanon (e-mail: jihad.fahs@lau.edu.lb). Aslan Tchamkerten and
 Mansoor I. Yousefi  are with the 
Communications and Electronics Department, T\'el\'ecom ParisTech,
Universit\'e Paris-Saclay, 75013 Paris, France (e-mails: \{aslan.tchamkerten,yousefi\}@telecom-paristech.fr).
}
}

%\markboth{IEEE Transactions on Information Theory}{}

\author{Jihad Fahs, Aslan Tchamkerten and Mansoor I. Yousefi}

\date{}

\maketitle

\begin{abstract}
The per-sample zero-dispersion channel model of the optical fiber is considered. 
It is shown that capacity is uniquely achieved by an input probability distribution that 
has  continuous  uniform  phase  and  discrete  amplitude  that  takes  on  finitely  many values.
%takes values over a finite number of concentric rings in the complex plane.
This result holds when the channel is subject to general input cost constraints, that include a peak amplitude
constraint and a joint average and peak amplitude constraint.

\end{abstract}

%\begin{IEEEkeywords}

%%%%%%%%%%%
% Section I: INTRODUCTION
%%%%%%%%%%% 

\section{Introduction}

Signal propagation in optical fibers can be modeled by the stochastic nonlinear Schr\"odinger (NLS)
equation~\cite{kramer2015upper}, capturing chromatic dispersion, Kerr
nonlinearity, and amplified spontaneous emission (ASE) noise.  Finding the capacity and the 
spectral efficiency of such a channel remains a formidable challenge, even in the special case of zero dispersion. The chief reason for this is that the channel is nonlinear. \iffalse
As a result, even the conditional probability distribution function (pdf) 
$p_{\vect{Y}|\vect{X}}(\vect{y}|\vect{x})$  of the channel output $\vect{y}\in\Complex^m$ given the channel input 
$\vect{x}\in\Complex^n$ in discrete-time models is unknown.\fi  To date the only non-asymptotic (in input power)  capacity result states that the capacity of 
the optical fiber is upper bounded by $\log(1+\snr)$, where \snr\ is the signal-to-noise ratio \cite{yousefi2015cwit,kramer2015upper}.  

In this paper we consider the per-sample zero-dispersion (PZD) channel model of the optical fiber. This 
channel is obtained by setting the dispersion coefficient to zero in the NLS equation and by sampling the output signal at the input signal bandwidth rate \cite[Sec. III]{yousefi2011opc}.
These simplifications yield a discrete memoryless channel that maps a complex input $x\in\Complex$ to a complex output
$y\in\Complex$ through a conditional probability density function (pdf) $p_{Y|X}(y|x)$. For this channel we show that capacity is uniquely achieved by an input signal that has continuous uniform phase and discrete amplitude that takes on finitely
many values independently of the phase. This result holds whenever the input is subject to a peak amplitude constraint, a joint peak amplitude 
and average cost constraint, or an average cost constraint with a cost function satisfying certain regularity conditions. 
This proves a conjecture made in \cite{yousefi2011opc} and shows that multi-ring modulation formats that are popular in optical fiber communication \cite{agrawal2007} are indeed optimal for the simplified PZD channel .\iffalse---potentially with geometric and probabilistic shaping.\fi

\subsection*{Related work}

The conditional pdf in the PZD model can be expressed as an infinite series
\cite{mecozzi1994llh,turitsyn2003ico,yousefi2011opc}.  The asymptotic capacity of the PZD model
 is $\frac{1}{2}\log \const P +o(1)$ as the average input power $\const P\rightarrow\infty$
\cite{turitsyn2003ico,yousefi2011opc}. Moreover, this asymptotic capacity is achieved by continuous pdfs, up to the $o(1)$ term \cite{mecozzi1994llh,turitsyn2003ico,yousefi2011opc}.

Since the work of Smith~\cite{smith1971ica}, several authors have established the discreteness of the capacity-achieving input distributions for a variety of channels and input constraints 
\cite{smith1971ica,Hirt88,shamai1995cap,Das,IA01,Aslan,chan2005cap,Gursoy2005,fahsj,Fahs2,fahs2017it}. 
Examples include complex additive white Gaussian noise channel \cite{shamai1995cap}, linear channels with additive noise~\cite{Das,Aslan,Fahs2,fahs2017it}, and  conditionally Gaussian channels~\cite{chan2005cap,fahsj,IA01,Gursoy2005}.

\subsection*{Contributions}
We first show that the capacity of the PZD channel is uniquely achieved and that the optimal input has a uniform phase that is independent of the amplitude. The main proof ingredient here is a symmetry argument.
%combined with a tightdiscrete approximation of the channel by a cascade of simpler (invertible) channels. 
In a second part we show that
the optimal amplitude takes on a finite number of values, following the methodology established by 
Smith \cite{smith1971ica}. Although the proof roadmap here is known, implementing it is far from straightforward; unlike channels considered so far where this proof technique is used (see, {\it{e.g.}}, \cite{smith1971ica,shamai1995cap,Das, Aslan,IA01,chan2005cap,Gursoy2005,fahsj,Fahs2,fahs2017it})
the PZD channel is non-additive---since the phase noise is a complex function of the input amplitude---and non-conditionally Gaussian. \iffalse This especially makes the analysis of the conditional entropy more difficult compared to that in~\cite{smith1971ica,Hirt88,shamai1995cap,Das,IA01,Aslan,chan2005cap,Gursoy2005,fahsj,Fahs2,fahs2017it}.\fi

The paper is organized as follows. In
Section~\ref{sec:pre} we recall the PZD channel model. In Section~\ref{sec:main} we present the main result and prove it in Section~\ref{sec:prems}. Finally, in Section~\ref{sec:con} we draw a few concluding remarks.

\subsection*{Notation}

Random  variables  and  their  realizations  are  denoted by  upper  and  lower  case  letters,  respectively. 
Real, non-negative real, integer, positive integer, non-negative integer, and complex numbers are respectively denoted by 
$\Reals$, $\Reals^+_0$,  $\mathbb Z$, $\Naturals$, $\Naturals_0$ and $\Complex$. 
The real and imaginary parts of $x \in \Complex$ are denoted by $\Re (x)$ and $\Im (x)$,  respectively.
The cumulative distribution function (cdf) and the pdf of a random variable $X$ are denoted by $F_X(x)$ and $p_{X}(x)$, respectively. \iffalse To simplify the notation, $F_X(x)$ is sometimes shortened to $F_X$.  Similarly,  $dF_X(x)$ is sometimes replaced with $dF(x)$, where $d$ is differential, \ie, $\int dF(x)=1$. \fi
The expected value of a random variable $X$ is $\E(X)$. 
The uniform distribution on the interval $[a,b)$
is denoted by $\U (a,b)$. Given two functions $f(x):\Reals\mapsto\Reals$ and $g(x):\Reals\mapsto\Reals$ we use the following standard order notations. We write $f(x)=\omega(g(x))$, or equivalently $g(x)=o(f(x))$, if for any $k>0$
there exists a $c>0$ such that $|f(x)|>k |g(x)|$ for all $|x|\geq
c$. We write $f(x)=\Omega(g(x))$ if there exists a $k>0$ and $c>0$ such that $|f(x)|>k |g(x)|$ 
for all $|x|\geq c$. Finally, we write $f(x) \equiv g(x)$ as $x \rightarrow x_0$ if $\lim_{x \rightarrow x_0} \frac{f(x)}{g(x)} = 1$.

%%%%%%%%%%%
% Zero-Dispersion Optical Fiber 
%%%%%%%%%%% 

\section{The Per-Sample Zero-Dispersion Channel}
\label{sec:pre}

%\textcolor{red}{Unclear which bandwitdh is given, which can be chosen (see W vs. W(0)). }
We briefly recall the PZD channel model derived in \cite[Sec.~III. A, Eq. 18]{yousefi2011opc}. %(see also \cite{turitsyn2003ico,mecozzi1994llh}).
The reader familiar with this model may move on to the conditional pdf \eqref{condpdf}.
 
Let $Q(t,z)$ denote the complex envelope of the signal as a function of time $t$ and distance $z$ along the
fiber. The propagation of the signal in the optical fiber is  described by a partial differential equation known as the stochastic NLS equation \cite[Chap.~2,3]{agrawal2007}. 
Setting the dispersion to zero in equation \cite[Eq.~7, with $\beta_2=0$]{kramer2015upper}, we obtain the following stochastic ordinary differential equation
\begin{IEEEeqnarray}{rCl}
\label{eq:zerodisp}
\frac{d Q(t,z)}{d z} = j\gamma Q(t,z) |Q(t,z)|^2 + N(t,z),\quad 0\leq
z\leq \mathcal L, 
\end{IEEEeqnarray}
where $\gamma$ is the nonlinearity coefficient, $\const L$ is the fiber length, and $j\defeq\sqrt{-1}$.
Furthermore, $N(t,z)$ is (zero-mean) circularly symmetric complex Gaussian noise satisfying
\begin{align*}
  \E \bigl(N(t,z)N^*(t',z')\bigr)=\sigma^2_0 \delta_{\const W}(t-t')\delta(z-z'),
\end{align*}
where $\sigma_0^2$ is the noise in-band power spectral
density, $\delta(x)$ is the Dirac delta function, and $\delta_{\const W}(x)\defeq \const W\sinc(\const W x)$, in which $\sinc(x)\defeq\sin(\pi
x)/(\pi x)$, and $\const W$ is the noise bandwidth.

Equation \eqref{eq:zerodisp} defines a continuous-time channel
from the input $Q(t,0)$ to the output $Q(t,\const{L})$. To obtain an equivalent discrete-time model we need to sample $Q(t,0)$ at rate $\const{W}(0)$ and $Q(t,\const{L})$ at rate $\const{W}(\mathcal{L})$,  where $\const W(z)$ 
denotes the signal bandwidth at distance $z$. Because of the nonlinearity, the signal bandwidth changes along the fiber and therefore in general $\const{W}(\mathcal{L})\ne\const{W}(0)$ \cite[Sec.~VIII]{yousefi2011opc}, \cite{kramer2017autocorrelation}.
How the bandwidth changes as a function of distance remains an important open problem in optical fiber communication. 
The PZD channel arises by assuming a sub-optimal receiver that samples the output signal at the input rate $1/\const{W}$, $\const{W}=\const{W}(0)$. 
This channel is discrete-time, memoryless, and stationary and maps, at rate $1/\const{W}$, an input $X\in \mathbb{C}$ to a random output $Y \in \mathbb{C}$.
%
%
\iffalse from input samples $(Q(k/\const{W},0))_{k\in\integers}$ to the corresponding output samples 
$(Q(k/\const{W},\const{L}))_{k\in\integers}$ at the output, where $k \in \mathbb{Z}$ is the time index. The sequence of random variables $(N(k/\const{W},z))_{k \in \mathbb{Z}}$ is \iid\ for any $z$, and independent of the input. 
It follows that the discrete-time channel is stationary and memoryless. Therefore, the input sequence $(Q(k/\const{W},0))_{k\in\integers}$ 
may be assumed to be \iid\ In what follows, we drop the time index $k$. 

We refer to the discrete memoryless channel $Q(0) \defeq Q(k/\const{W},0) \mapsto Q(\const{L}) \defeq Q(k/\const{W},\const{L})$ as the PZD channel, where $Q(0)\in \mathbb{C}$ is the channel input and $Q(\const L) \in \mathbb{C}$ is the channel output.\fi
%
%
The input output relation is obtained by solving \eqref{eq:zerodisp} \cite[Eq.~30]{yousefi2011opc}
\begin{align}\label{eq:zerodiso}
Y=[X+W(\const L)]e^{j \gamma \int_0^\const L |X+W(\ell)|^2 d\ell},
\end{align}
where $W(\ell)$ is a complex Wiener process with variance $\ell \sigma^2 $, $\sigma^{2}=\sigma_{0}^{2}\mathcal{W}$.
Letting $(R_0,\Phi_0)$ and  $(R,\Phi)$ denote the polar coordinates of 
$X$ and $Y$, respectively, equation \eqref{eq:zerodiso} equivalently defines the conditional pdf (see \cite{mecozzi1994llh}, \cite{turitsyn2003ico}, \cite[Eq.~18]{yousefi2011opc})
 \begin{equation}
\label{condpdf}
p_{R,\Phi|R_0,\Phi_0}(r,\phi | r_0,\phi_0) \defeq p(r,\phi | r_0,\phi_0) =  \frac{1}{2 \pi} p_{R|R_0}(r|r_0) + \frac{1}{\pi}\hspace{-0.05cm}\sum^{+\infty}_{m =1}\hspace{-0.05cm}\Re\left(C_{m}(r,r_0)e^{jm(\phi-\phi_0-\gamma r^2_0\mathcal{L})}\right),
\end{equation}
 where 
\begin{equation}
\label{condR}
p_{R|R_0}(r|r_0) \defeq \frac{2r}{\sigma^2\mathcal{L}}e^{-\frac{r^2+r_0^2}{\sigma^2 \mathcal{L}}} I_0\left(\frac{2 r r_0}{\sigma^2 \mathcal{L}}\right),
\end{equation} 
\begin{equation}
\label{coeffC}
C_m(r,r_0) \defeq r b_me^{-a_m(r^2+r_0^2)}I_{m}\left(2b_mr_0r\right),
\end{equation}
where $I_m(\cdot)$ is the modified Bessel function of the first kind and where
\begin{eqnarray}
a_m &\defeq& \frac{\sqrt{jm\gamma}}{\sigma}\coth\left(\sqrt{jm\gamma\sigma^2}\mathcal{L}\right), \label{coeffa}\\
b_m &\defeq & \frac{\sqrt{jm\gamma}}{\sigma}\frac{1}{\sinh\left(\sqrt{jm\gamma\sigma^2}\mathcal{L}\right)}. \label{coeffb}
\end{eqnarray}
Note that, because of the potentially sub-optimal discretization of the output, the capacity of the PZD channel \eqref{condpdf} (measured in bits/s) is at most equal to the capacity of the continuous-time zero-dispersion channel \eqref{eq:zerodisp} \cite{yousefi2011opc}. 

\section{Main Result}
\label{sec:main}

We consider the PZD channel \eqref{condpdf} when the input $(R_0,\Phi_0)$ is subject to one of the following constraints.
\begin{description}
\item[Peak amplitude constraint:] 
The cdf of $R_0$ belongs to the set 
\begin{IEEEeqnarray}{rCl}
{\mathcal{P}} &\defeq& \Bigl\{F_{R_{0}}(r_0) \,\, : \,\,\int _{0}^{\rho} dF_{R_0}(r_0) =1\Bigr\},
\label{acos}  
\end{IEEEeqnarray}
for some $0< \rho<\infty$. 

\item[Average cost constraint.] The cdf of $R_0$ belongs to the set
\begin{IEEEeqnarray}{rCl}
  \mathcal{A} &\defeq & \Bigl\{F_{R_{0}}(r_0) ~:~ \int _{0}^{+\infty}
\mathcal{C}(r_0) \, dF_{R_0}(r_0) \leq A \Bigr\},
\label{ccos}
\end{IEEEeqnarray}
for some $0<A<\infty$ and cost function $\mathcal{C}(r_0): \Reals_0^+\mapsto \Reals_0^+$ that satisfies the  following conditions: 
\begin{itemize}
\item[C1.] $\mathcal{C}(r_0)$ is lower semi-continuous,  non-decreasing, $\mathcal{C}(0) = 0$ and $\lim_{r_0 \rightarrow +\infty}\mathcal{C}(r_0) = +\infty$;
  \item[C2.] $\mathcal{C}(r_0)$ can be analytically extended from $[0,\infty)$ to a horizontal strip in the complex plane
  \begin{IEEEeqnarray*}{rCl}
\mathcal{S}_{\delta} = \Bigl\{z \in \Complex ~:~
  \left|\Im(z)\right| < \delta, \:,\:\delta>0 \Bigr\},     
    \end{IEEEeqnarray*}
for some $\delta>0$;
  %a horizontal strip in the complex plane
  %\begin{IEEEeqnarray*}{rCl}
%\mathcal{S}= \Bigl\{z \in \Complex ~:~ \Re(z) \geq 0 \:,\:
  %\left|\Im(z)\right| < \delta, \:,\:\delta>0 \Bigr\},     
    %\end{IEEEeqnarray*}
%for some $\delta>0$;
\item[C3.]
$\mathcal{C}(r_0)=\omega(r^2_0)$.
\end{itemize}

\item[Joint peak amplitude and average cost constraint.] The cdf of $R_0$ belongs to the set ${\mathcal{P}}\cap {\mathcal{A}}'$, where ${\mathcal{A}}'$ is defined as ${\mathcal{A}}$ but with condition C3 replaced by the weaker condition C3: $\mathcal{C}(r_0)=\omega(\ln r_0)$.

\end{description}

An example of a family of cost functions satisfying conditions C1-C2 is $\mathcal C(r_0)=r_0^q$  for $q \in \Naturals$. 

%\begin{remark}
%The sets ${\mathcal{P}}$, ${\mathcal{A}}$, and ${\mathcal{A}}'$  depend on $\rho$ and $A$. To %simplify notation and since our main result (stated hereafter) holds irrespectively of the actual %values of these parameters we leave out any explicit reference to them.
%\end{remark}

Abusing slightly notation, we will use $X$ and $Y$ to denote both the complex values and their representations in polar coordinates. With this convention, the channel capacity is 
\begin{IEEEeqnarray}{rCl}
\label{eq:cap}
C =  \sup_{F_{X}: F_{R_{0}}(r_0) \in \mathcal{F}}I\left(X;  Y\right), 
%= \sup_{F_{\mathbf T_0}: F_{R_{0}}(r_0) \in \mathcal{F}} I\left(F_{\mathbf T_0}\right) = \sup_{F_{\mathbf T_0}: F_{R_{0}}(r_0) \in \mathcal{F}} \left(h(\mathbf T) - h(\mathbf T|\mathbf T_0)\right),
\end{IEEEeqnarray}
where $I\left(X; Y\right)  $ denotes the mutual information and $\mathcal{F}$ represents any of the sets $\mathcal{P}$, $\mathcal{A}$, or 
$\mathcal{P} \cap \mathcal{A}'$.

\begin{theorem}
\label{mnth} 
%The capacity of the channel \eqref{condpdf} subject to one of the input constraints given by %$\mathcal{P}$, $\mathcal{A}$, or 
%$\mathcal{P} \cap \mathcal{A}'$ is finite and 
The channel capacity in \eqref{eq:cap} is finite and is achieved by a unique cdf $F_{R^*_0,\Phi^*_0}(r_0,\phi_0)$ where the phase $\Phi^*_0$ is uniform over $[0,2\pi)$ and where the amplitude $R^*_0$ takes on a finite number of values independently of $\Phi^*_0$. 
\end{theorem}

Hence, under fairly general conditions the support of the optimal input consists of a finite number of concentric rings in the complex plane. Without
a peak constraint, the same result holds provided that the cost function grows faster than $r_0^{2}$. Whether the conclusion extends to a strict average power constraint, that is with cost function $\const C(r_0)=r_0^2$, remains an open problem.

\begin{remark}
\label{rem:multiple}
Theorem~\ref{mnth} still holds when the input is subject to a finite number of constraints given by cost functions $\mathcal{C}_i(r_0)$, $1\leq i\leq M<\infty$, satisfying C1--C2, and such that $\const C_i(r_0)=\omega(r^2_0)$ for at
least one cost constraint. For example, we may consider a joint second and a fourth moment
constraint as in the non-coherent Rician fading channel \cite{Gursoy2005}. 
\end{remark}

%%%%%%%%
%% SECTION IV: PROOF
%%%%%%%%

\section{Proof of Theorem~\ref{mnth}}
\label{sec:prems}
To prove Theorem~\ref{mnth} we first show that capacity is finite and achieved by a 
unique cdf. Then, we show that the optimal input phase is uniform and independent 
of the optimal input amplitude. Finally, we prove that the optimal input amplitude takes on a finite
number of values.

\begin{lemma}
The capacity \eqref{eq:cap} is finite and achieved by a unique cdf.

%\textcolor{red}{[Technically $\cal{F}$ is a probability space over $\Reals$ not over $\Reals^2$. Here we are abusing the definition. It is the same when %we talk about the closure of \cal{F} in $\Reals^2$. I believe we should add a statement, maybe after the constraints saying that when we refer to %$\cal{F}$ we will be referreing to the space pf probability in $\Reals^2$ induced by $\cal{F}$.]}

\label{lemm:sup}
\end{lemma}
\begin{proof}  
It is known that that the  the set of 
input distributions satisfying an input constraint in Section \ref{sec:main} is compact \cite[Thm.~3]{FAF15}, \cite[Prop.~1]{smith1971ica}. 
Further, in Appendix~\ref{app:cont} we prove that  mutual information $I(X;Y)$ is continuous in
$F_{X}$. 
Therefore, from the extreme value theorem \cite[Thm. 4.16]{rudin1964principles}, the supremum in \eqref{eq:cap} is finite and achieved.

For uniqueness, it suffices to prove that $I(X;Y) = h(Y) - h(Y|X)$ is a strictly concave function of $F_{X}$.  
The term $h(Y|X)$ as a function of $F_{X}$ is linear hence concave. The term $h(Y)$ is strictly concave in $F_{Y}$. Below, we prove that the linear mapping $F_{X}\mapsto F_{Y}$ is injective which then implies that $h(Y)$, and therefore $I(X;Y)$, is strictly concave in~$F_{X}$. 

Changing variable $z'=\const{L}-z$ in \eqref{eq:zerodisp}, $X$ is obtained from $Y$ according to \eqref{eq:zerodisp} with
$\gamma\mapsto -\gamma$ and $N(t,z)\mapsto -N(t,\const{L}-z')$. Since $-N$ and $N$ are identically distributed, we obtain the symmetry relation $p_{X|Y}(x|y;\gamma) = p_{Y|X}(x|y;-\gamma)$. The channel is therefore
invertible and injective.
\end{proof}

\subsection{Optimal input phase}

\label{sec:phase-uniform}

In this section we show that the optimal input phase is uniform, independently  of the optimal input amplitude.

\begin{lemma}
\label{lem:chansym}
The following properties hold:
\begin{enumerate}
\item the output amplitude $R$ is independent of the input phase $\Phi_0$;
\item the capacity-achieving input in \eqref{eq:cap} has uniform phase
  $\Phi_0^*\sim \U (0,2\pi)$ independent of the amplitude $R_0^*$, \ie, 
\begin{equation*}
\label{inrotinv}
dF_{(R^*_0,\Phi_0^{*})}(r_0,\phi_0) = \frac{1}{2\pi}d\phi_0dF_{R^*_0}(r_0).
\end{equation*} 
\end{enumerate}
%\label{lem:phi0-uniform}
\end{lemma}
\begin{proof}
\emph{1)} From \eqref{eq:zerodiso}
\begin{align}
  R&=\left|R_0e^{j\Phi_0}+W(\const{L})\right|\nn\\
&=\left|R_0+W'(\const{L})\right|\nn\\
&\overset{d}{=}\left|R_0+W(\const{L})\right|\nn,
\end{align}
where $W'(\const L)\defeq e^{-j\Phi_0} W(\const{L})$. 
The last equality is in distribution where we used the fact that $W'(\const L)$ and
$W(\const{L})$ are identically distributed from the circularly symmetry property. 

\emph{2)} From \eqref{eq:zerodiso} and by the circularly symmetry of the complex Wiener process, if output $Y$ corresponds to input $X$ then output $e^{j\theta}Y$ corresponds to input $e^{j\theta}X$, for any fixed $\theta \in [0,2\pi)$. 
Therefore, the mutual information is invariant under an input rotation.  
Hence, if $X^*$ is capacity-achieving, so is 
$e^{j\theta} X^*$ for any $ \theta \in [0,2\pi)$. By Lemma~\ref{lemm:sup}, the capacity-achieving input distribution is unique.
%\begin{align*}
% \label{usta}X^*&\overset{\text{d}}{=} e^{j\theta} X^*,
%\end{align*}
%  for any $ \theta \in [0,2\pi)$. 
This implies that $X^*$ has a uniform phase $\Phi_0^*\sim\U(0,2\pi)$, that is independent of the amplitude.

\end{proof}

Lemmas~\ref{lemm:sup} and~\ref{lem:chansym} yield:
\begin{proposition}
\label{cormax}
The capacity \eqref{eq:cap} simplifies to
\begin{IEEEeqnarray}{rCl}
\label{chancap}
C &=& 
\max_{F_{R_{0}} \in \mathcal{F}}  I\left(F_{R_0}\right),
\end{IEEEeqnarray}
 where
 \begin{IEEEeqnarray*}{rCl}
I\left(F_{R_0}\right) &\defeq&  I\left(R_0,\Phi^*_0;R,\Phi\right)\nonumber \\ 
%&=& h\left(R,\Phi \right) -h\left(R,\Phi|R_{0},\Phi^*_{0}\right) \nonumber\\
& =& h\left(R\right) + \ln (2 \pi) - h\left(R,\Phi|R_{0},\Phi^*_{0}\right) \label{eq:mutsim},
 \end{IEEEeqnarray*}
with  $\Phi^*_{0}\sim\U (0,2\pi)$. 

\end{proposition}

\iffalse\begin{proof}
From Lemma~\ref{lemm:sup}, the capacity is achievable by the maximum of the mutual information. 
From Lemma~\ref{lem:chansym}, $\Phi$ is uniform when $\Phi_0$ is uniform. However, note that,
conditioned on $R_0$ and $\Phi_0^*$,
$\Phi$ is not uniform. Therefore, phase is not completely simplified  in $I(F_{R_0})$.

\end{proof}
\fi
From Proposition~\ref{cormax}, finding capacity reduces to maximizing the strictly concave function $I(F_{R_0})$. The set $\mathcal{F}$ is linear, and thus weakly differentiable, in $F_{R_0}$.
Extending a result in \cite{fahs2017it} for additive noise channels to the PZD channel, 
it can be shown that $I(F_{R_0})$ is weakly differentiable.
As a result, the KKT conditions characterize the optimal input $F_{R_0^*} $ in terms of necessary and sufficient conditions.

\subsection{KKT conditions.}
The derivation of the KKT conditions can be found in 
\cite[Corollary 1]{smith1971ica} for the peak amplitude constraint, and in 
\cite[Theorem 15]{fahsj} for the average cost constraint. We present their final forms   below. 
\begin{description}
\item[Peak amplitude constraint.]

An input amplitude $R_0^*$ with cdf $F_0^*\in\mathcal{P}$  
 achieves the capacity $C$ in \eqref{chancap} if and only if
\begin{IEEEeqnarray}{rCl}
   \text{LHS}_{\rho}(r_0) &\overset{\text{def}}{=}&   C  - \ln(2\pi) + \int^{+\infty}_{0} p
   \left(r|r_0\right) \ln p(r;F_{0}^{*})\,dr  + \frac{1}{2 \pi} \int_{0}^{2 \pi} h\left(R,\Phi|r_{0},\phi_{0}\right) \, d\phi_0
   \nn
   \\
   &\geq& 0,
  \label{eq:KKT1}
\end{IEEEeqnarray}
for all $0 \leq r_0 \leq \rho$, with equality if $r_0$ is a
point of increase of $F_{0}^*$.

\item[Average cost constraint.]

An input amplitude $R_0^*$ with cdf $F_0^*\in\mathcal{A}$ achieves the capacity $C$ in 
\eqref{chancap} if and only if there exists $\nu \geq 0$ such that
\begin{IEEEeqnarray}{rCl}
 \text{LHS}_{A}(r_0)  &\overset{\text{def}}{=}&   C - \ln(2\pi) +
 \int^{+\infty}_{0} p \left(r|r_0\right) \ln p(r;F_{0}^{*})\,dr + \nu (\mathcal{C}(r_0) - A) + \frac{1}{2 \pi} \int_{0}^{2 \pi} h\left(R,\Phi|r_{0},\phi_{0}\right) \, d\phi_0 
 \nn
 \\
 &\geq& 0,
  \label{eq:KKT} 
\end{IEEEeqnarray}
for all $r_0 \geq 0$, with
equality if $r_0$ is a point of increase
of $F_{0}^*$.

\item[Joint peak amplitude and average cost constraint.]

The KKT condition in this case is same as (\ref{eq:KKT}), but 
for all $0 \leq r_0 \leq \rho$.  

\end{description}

\subsection{Discreteness of optimal input}
Following Smith~\cite{smith1971ica}, we apply the identity theorem~\cite[Thm.~10.26]{Sil} to the KKT conditions.
%argument, we show that the KKT conditions are satisfied by distributions with finite support. 
The first part of the argument is to show that $ \text{LHS}_{\rho}(r_0)$ and  $\text{LHS}_{A}(r_0)$ in \eqref{eq:KKT1}--\eqref{eq:KKT}
can be analytically extended from $r_0\in \Reals^+_0$  to an open connected region $\mathcal{O} _\delta$, $\delta >0$,  
%(which is considered to be both open and closed)} 
in the complex plane. This is established in Lemma~\ref{lem:analyticity} in Appendix~\ref{app:analyticity}. The second part of the argument consists in deriving the optimality of discrete inputs by way of contradiction in the identity theorem. 
\begin{description}

\item[Peak amplitude constraint.]

Suppose that the points of increase of $F_0^*$ have an accumulation point in the interval $[0,\rho]$. 
Then, since $\text{LHS}_{\rho}$($z$) is analytic on $\mathcal O _\delta$, it is identically zero on
$\mathcal O _\delta$ from the identity theorem \cite[Thm.~10.26]{Sil}. 
Therefore the KKT condition (\ref{eq:KKT1}) is satisfied with
equality for all $r_0 \geq 0$. 

Using the upper bound on $\text{LHS}_{\rho}(r_0)$ in Lemma~\ref{lemupplhs} in 
Appendix~\ref{app1}, there exists $K > 0$ such that 
\begin{IEEEeqnarray}{rCl}
 0 &=& \text{LHS}_{\rho}(r_0) \nn\\ 
 &\leq&  C + \ln \left(\frac{1}{K}\right) + \frac{r_0^2}{\sigma^2 \mathcal{L}}- \ln\left(1-\xi(r_0)\right) + \left(\rho - (1-\epsilon)r_0\right)\sqrt{\frac{\pi}{\sigma^2 \mathcal{L}}} \text{L}_{\frac{1}{2}}\left(-\frac{r_0^2}{\sigma^2 \mathcal{L}}\right), 
 \label{acontrad} 
\end{IEEEeqnarray}
where $\text{L}_{\frac{1}{2}}(\cdot)$ is a Laguerre polynomial, $\xi(r_0) \rightarrow 0$ as $r \rightarrow +\infty$, 
and $\epsilon\in (0,1)$. Since
 $\text{L}_{\lambda}(x) \equiv
 \frac{|x|^{\lambda}}{\Gamma(1+\lambda)}$ as $x \rightarrow
 -\infty$~\cite{abra1964}, we have  
 \begin{equation}
 \lim_{r_0 \rightarrow +\infty}\frac{r_0 \sqrt{\frac{\pi}{\sigma^2 \mathcal{L}}} \text{L}_{\frac{1}{2}}\left(-\frac{r_0^2}{\sigma^2 \mathcal{L}}\right)}{r_0^2} = \frac{2}{\sigma^2 \mathcal{L}}.
 \label{idenlag}
 \end{equation}
Dividing \eqref{acontrad} by $r_0^2>0$ and taking the limit as $r_0 \rightarrow +\infty$, we obtain
 \begin{equation}
 0 \leq \frac{1-2(1-\epsilon)}{\sigma^2 \mathcal{L}},\quad \forall
 \epsilon\in (0,1),
 \label{contrapr}
 \end{equation}
 where we used the fact that $\lim_{r_0
   \rightarrow +\infty}\xi(r_0) = 0$ (see Lemma~\ref{lembds}). Since (\ref{contrapr}) holds
 for any $0 < \epsilon < 1$, choosing $\epsilon < \frac{1}{2}$ yields
 a contradiction. 
 
 It follows that the points of increase of $F^*_{0}$ are isolated. Since the interval $[0,\rho]$ is bounded, 
 there are a finite number of mass points.

\item[Average cost constraint.]

Consider the analytic extension $\text{LHS}_A(z)$ of $\text{LHS}_A(r_0)$ in $\mathcal O _\delta$ and 
suppose that $F_0^*$ has an infinite number of points of increase in a
bounded interval in $\Reals^{+}_0$. 
Then, from the identity theorem, $\text{LHS}_A(z)$ is identically zero on $\mathcal{O}_\delta$. 
 As a result, the KKT condition \eqref{eq:KKT} is satisfied with equality for all $r_0 \geq 0$ and we have
\begin{IEEEeqnarray}{rCl}
 0 &=& \text{LHS}_{A}(r_0)
 \nn
 \\
 &>& \nu (\mathcal{C}(r_0) - A) + C + \ln\left(\frac{k_1}{2 \pi k_u}\right) + \frac{1}{\sigma^2 \mathcal{L}} r_0^2  - r_0 \sqrt{\frac{\pi}{\sigma^2 \mathcal{L}}} \text{L}_{\frac{1}{2}}\left(-\frac{r_0^2}{\sigma^2 \mathcal{L}}\right), 
 \label{contrad}
 \end{IEEEeqnarray}
 where \eqref{contrad} follows from the lower bound on $\text{LHS}_A(r_0)$ in Lemma~\ref{lemlowlhs} in 
 Appendix \ref{app1}. Dividing \eqref{contrad} by $r_0^2>0$ and taking the limit $r_0 \rightarrow +\infty$ gives
\begin{equation}
\label{eqlimub}
\nu \lim_{r_0 \rightarrow +\infty} \frac{\mathcal{C}(r_0)}{r_0^2} - \frac{1}{\sigma^2 \mathcal{L}}  < 0,
\end{equation}
 where we used \eqref{idenlag}. The inequality \eqref{eqlimub} is impossible for $\mathcal{C}(r_0) = \omega(r^2_0)$,
 unless $\nu = 0$ which is ruled out in \cite[Lemma 5]{fahs2017it}.
 
We now prove that the number of mass points is finite. Suppose that $F_0^*$ has an 
infinite number of points of increase
with only a finite number of them in any bounded interval in
$\Reals^{+}_0$. Then, the points of increase tend
to infinity. We establish in Lemma~\ref{lemlowlhs} in Appendix~\ref{app1} a lower bound on $\text{LHS}_{A}(r_0)$ which diverges to $+\infty$ as $r_0 \rightarrow +\infty$. This implies that $ \text{LHS}_{A}(r_0) \neq 0$ at large values of  $r_0$ which contradicts the possibility of having arbitrarily large 
mass points. It follows that  $F_0^*$ has a finite number of isolated points of increase in $\Reals^{+}_0$.
  
The above proof immediately generalize to
 the case with multiple average cost constraints, described in Remark \ref{rem:multiple}.

\item[Joint average cost and peak amplitude constraints.]

Suppose that $F_0^*$ has an infinite number of points of increase in 
$[0,\rho]$. Then the KKT condition \eqref{eq:KKT} implies that
\begin{equation}
\text{LHS}_{A}(r_0) = 0, \quad \forall r_0 
\geq 0.
\end{equation}
Thus \eqref{eqlimub} holds true and 
\begin{equation}
\label{eqlimub1}
\nu \lim_{r_0 \rightarrow +\infty} \frac{\mathcal{C}(r_0)}{r_0^2} < \frac{1}{\sigma^2 \mathcal{L}}.
\end{equation}

On the other hand, since the support of $F_0^*$ is restricted to
$[0,\rho]$, we can apply the upper bound in Lemma~\ref{lemupplhs} to the KKT condition \eqref{eq:KKT} to obtain
\begin{align*}
 \nu \left(\mathcal{C}(r_0) - A\right) + C  + \ln \left(\frac{1}{K}\right) + \frac{r_0^2}{\sigma^2 \mathcal{L}} -  \ln\left(1-\xi(r_0)\right) \nonumber + \left(\rho - (1-\epsilon)r_0\right)\sqrt{\frac{\pi}{\sigma^2 \mathcal{L}}} \text{L}_{\frac{1}{2}}\left(-\frac{r_0^2}{\sigma^2 \mathcal{L}}\right)\geq 0.
\end{align*}
Dividing by $r^2_0>0$ and letting $r_0 \rightarrow +\infty$ gives
\begin{equation*}
\nu \lim_{r_0 \rightarrow +\infty} \frac{\mathcal{C}(r_0)}{r_0^2} \geq
\frac{2(1-\epsilon)-1}{\sigma^2 \mathcal{L}},\quad \forall \epsilon\in(0,1).
\end{equation*}
Taking the limit $\epsilon \rightarrow 0$
\begin{equation}
\label{eqlimlb}
\nu \lim_{r_0 \rightarrow +\infty} \frac{\mathcal{C}(r_0)}{r_0^2} \geq \frac{1}{\sigma^2 \mathcal{L}}.
\end{equation}
Equations~(\ref{eqlimub1}) and~(\ref{eqlimlb}) establish a contradiction. 
It follows that $F_0^*$ has a finite number of isolated points of increase in $[0,\rho]$.
\end{description}

%%%%%%%%
%% SECTION V: CONCLUSIONS
%%%%%%%%

\section{Conclusions}
\label{sec:con}

We studied the capacity-achieving input distribution in the per-sample zero-dispersion model of the optical fiber, subject
to a peak amplitude, an average cost, or a joint cost constraint. 
We proved that the capacity-achieving input distribution in this channel is unique, has a uniform phase 
that is independent of the amplitude, and the distribution of the amplitude is discrete with a finite number of mass points. In other words, multi-ring modulation formats commonly used in optical communication can achieve channel capacity---potentially with non-uniform ring spacing and probabilities.
Whether such constellations are optimal for the dispersive optical fiber is an interesting open problem.

%%%%%%%%
%% APPENDICES
%%%%%%%%

\appendices

\section{Continuity of the Mutual Information}
\label{app:cont}

In this appendix we prove that $I(X;Y)$ in the channel \eqref{condpdf} is continuous in $F_{X} \in \mathcal{F}$.
%(as the L\'evy-Prohorov's metric metrizes weak convergence). 
First, we establish upper and lower bounds on the conditional pdf in Lemma~\ref{lembds}, upper and lower bounds on the output pdf in Lemma~\ref{lem:outbds}, and an upper bound on the conditional entropy density in Lemma~\ref{lemm:entropy-density}. 
These bounds are then used to prove the continuity of the conditional entropy in Lemma~\ref{thm:cont-cond-entr} and the continuity of the output entropy in Lemma~\ref{thm:cont-out-entr}. 
Lemmas~\ref{thm:cont-cond-entr} and~\ref{thm:cont-out-entr} yield the desired result.

Note that, since the function $i(p)=-p\ln(p)$ has opposite signs when $0<p<1$ and $p>1$, to upper-bound entropies,
we need both upper and lower bounds on the probability distributions.

We begin by recalling a few properties of the modified Bessel function.

\begin{lemma}[Bounds on the Bessel Functions]
\label{lem:bessel}
The modified Bessel functions of the first kind $I_m(x)$ satisfy the following properties:
\begin{enumerate}
\item $I_0(x)\geq 1$ when $x\geq 0$, with equality iff $x=0$. Furthermore,  $I_m(x)\geq 0$ when $x\geq 0$,
$m\in\Naturals_0$;
\item $|I_m(z)| \leq |I_0(\Re(z))|$ when $z \in \Complex$, $m \in \Naturals_0$. Furthermore, $I_m(x)\leq e^{|x|}$ when $x\in\Reals$, $m\in\Naturals_0$;

    \item If $m \in \Naturals_0$ and  $z \rightarrow 0$
\begin{eqnarray*}
I_m(z) \equiv \frac{\left(\frac{1}{2} z\right)^m}{m!};
\label{eq:Im<I0}
\end{eqnarray*}

\item For any $0 < \epsilon < 1$, there exists a $K > 0$ such that $I_0(x) \geq Ke^{(1-\epsilon)x}$, when $x \geq 0$;

\item $I_m(x)$ is monotonically increasing for $x\geq 0$ and $m\in\Naturals_0$;

\item $I_m(z)$ is an analytic function of $z$ in the entire complex plane for $m \in \mathbb{Z}$. 
\end{enumerate}
\end{lemma}

\begin{proof}
The proof of the well-known inequalities in 1) and 2) is straightforward using the integral definition of  
$I_m(x)$, $m \in \Naturals_0$~\cite[Prop. 9.6.19]{abra1964}. Property 3) can be found in~\cite[Prop. 9.6.7]{abra1964}. For the inequality in 4), note that
$e^{(\epsilon-1) x} I_0(x) \equiv \frac{e^{\epsilon x}}{\sqrt{2 \pi x}}$ as $x\rightarrow +\infty$~\cite[Prop. 9.7.1]{abra1964}. Therefore the
positive and continuous function $e^{(\epsilon-1) x} I_0(x)$ is lower bounded by some $K>0$. The proof of 5) 
is due to 1) in this Lemma, and because $dI_m(x)/dx = I_{m+1}(x) + \frac{m}{x} I_m(x)$~\cite[Prop. 9.6.26]{abra1964}. Property 6) can be found in~\cite[Prop. 9.6.1]{abra1964}.
\end{proof}

\begin{lemma}[Bounds on the Conditional pdf]
\label{lembds}

 The conditional pdf \eqref{condpdf} satisfies the following inequalities.

\begin{enumerate}
\item Upper bound:
  \begin{IEEEeqnarray}{rCl}
    p(r,\phi | r_0,\phi_0) < k_{u}p_{R|R_0}(r|r_0),
    \label{eq:cond-pdf-ub}
  \end{IEEEeqnarray}
where 
\begin{eqnarray}
k_{u} \defeq \frac{1}{2 \pi} \left(1+\sqrt{2} \sum^{+\infty}_{m =1} \frac{ \beta_{m} }{\sinh(\beta_{m})}\right) < \infty,
\label{eq:ku}
\end{eqnarray} 
in which $\beta_{m} \defeq \sqrt{\frac{m \gamma}{2}} \sigma \const{L}$, $m > 0$.
The conditional pdf of the amplitude is upper bounded as
\begin{IEEEeqnarray}{rCl}
p_{R|R_0}(r|r_0)&\leq&
\frac{2r}{\sigma^2\const{L}}e^{-\frac{(r-r_0)^2}{\sigma^2\const{L}}}
\label{eq:p(r|r0)-up1}
\\
&\leq&
\frac{2r}{\sigma^2\const{L}}.
\label{eq:p(r|r0)-up2}
\end{IEEEeqnarray}

\item Lower bound:
  \begin{IEEEeqnarray}{rCl}
p(r,\phi | r_0,\phi_0) \geq \frac{1}{2 \pi}
p_{R|R_0}(r|r_0)\left(1-\xi(r_0)\right),
\label{condpdflbf}
  \end{IEEEeqnarray}
where
\begin{equation}
\xi(r_0) = \sqrt{2}e^{-\left(\Re(a_1)-\frac{1}{\sigma^2 \mathcal{L}}\right) r_0^2}\sum^{+\infty}_{m =1} \frac{\beta_{m} }{\sinh(\beta_{m})}.
\label{eq:zeta-r0}
\end{equation}
We have $\xi(r_0) \rightarrow 0$ as $r_0 \rightarrow +\infty$.

\end{enumerate}
%\label{lem:pdf-bounds}
\end{lemma}

\begin{proof}

We apply the following inequalities to the conditional pdf \eqref{condpdf}
\begin{align}
\label{eq:up-low-cm}
-\left|C_{m}(r,r_0)\right| 
\leq
\Re\left(C_{m}(r,r_0) e^{jm(\phi-\phi_0-\gamma r_0^2\const{L})}\right) 
\leq
\left|C_{m}(r,r_0)\right|,
\end{align}
where the expression of $C_m(r,r_0)$ is given by equation \eqref{coeffC}. First, we upper-bound $C_m(r,r_0)$ as follows
\begin{IEEEeqnarray}{rCl}
\bigl|C_m(r,r_0)\bigr|&=&\left| r b_me^{-a_m(r^2+r_0^2)}I_{m}\left(2b_mr_0r\right)\right|
\nn
\\
&\leq&
 r|b_m|e^{-\Re(a_m)(r^2 + r_0^2)}I_0\left(2\Re(b_m)r_0r\right). 
\label{eq:cm-ub}
\end{IEEEeqnarray}
We upper-bound $|b_m|$ and $\Re(b_m)$, and lower-bound $\Re(a_m)$, in \eqref{eq:cm-ub}.

Using the expression for $a_m$ in (\ref{coeffa}), we have $\Re(a_m) \defeq \frac{1}{\sigma^2 \mathcal{L}}t(\beta_m)$ where $\beta_{m}\defeq \sqrt{\frac{m \gamma}{2}} \sigma \mathcal{L}$, $m > 0$, and where
\begin{eqnarray}
t(x) = \frac{x \left(\sinh(2x) + \sin(2x)\right)} {2\left(\sinh^2(x) + \sin^2(x)\right)}. \label{teq}
\end{eqnarray}
It can be verified that $t(x)$ is increasing for $x > 0$ and $\lim_{x \rightarrow 0} t(x) = 1$. Thus, 
$t(\beta_m)\geq t(\beta_1)>t(\beta_0)=1$ when $m\geq 1$, and we obtain  two lower bounds
\begin{eqnarray}
\Re(a_m) &\geq&
\Re(a_1)
\label{lwbnda1}
\\
&>&
\frac{1}{\sigma^2 \mathcal{L}}.
\label{lwbda}
\end{eqnarray}

Similarly, using the expression for $b_m$ in (\ref{coeffb}), we obtain
\begin{eqnarray}
\Re(b_m) &=& \frac{\beta_{m}}{\sigma^2\mathcal{L}}\frac{\sinh(\beta_{m}) \cos(\beta_{m})+ \text{cosh}(\beta_{m})\sin(\beta_{m})}{\sinh^2(\beta_{m}) + \sin^2(\beta_{m})}
\nonumber
\\
&\leq& \frac{1}{\sigma^2 \mathcal{L}}, \label{unibdb} 
\end{eqnarray}
and 
\begin{IEEEeqnarray}{rCl}
|b_m| &=& \frac{\sqrt{2} \beta_{m}}{\sigma^2 \const{L} \sqrt{\sinh^2(\beta_{m}) + \sin^2(\beta_{m})}} \hspace{-0.05cm} 
\nonumber
\\
&\leq& \frac{\sqrt{2} \beta_{m} }{\sigma^2 \const{L} \sinh(\beta_{m})}.
\label{unibdb1}
\end{IEEEeqnarray}

Substituting \eqref{lwbnda1}--\eqref{unibdb1} into \eqref{eq:cm-ub}, we obtain two
upper bounds on $C_m(r,r_0)$
\begin{IEEEeqnarray}{rCl}
|C_m(r,r_0)|&\leq&
\frac{\sqrt{2} r}{\sigma^2 \const{L}}
\frac{ \beta_{m} }{\sinh(\beta_{m})} 
e^{-\Re(a_1)(r^2 + r_0^2)}  I_0\left(\frac{2 r_0r}{\sigma^2 \mathcal{L}}\right)
\label{eq:cm-ub-1}
\\
&< &
\frac{\sqrt{2} r}{\sigma^2 \const{L}}
\frac{ \beta_{m} }{\sinh(\beta_{m})} 
e^{-\frac{r^2 + r_0^2}{\sigma^2\const{L}}}  I_0\left(\frac{2 r_0r}{\sigma^2 \mathcal{L}}\right).
\label{eq:cm-ub-2}
\end{IEEEeqnarray}

\emph{Upper Bound.} 

Applying the upper bound in \eqref{eq:up-low-cm} to the conditional pdf \eqref{condpdf} and using the 
second upper bound on $C_m(r,r_0)$ in \eqref{eq:cm-ub-2}
\begin{IEEEeqnarray*}{rCl}
p(r,\phi | r_0,\phi_0) &\leq& 
\frac{1}{2 \pi} p_{R|R_0}(r|r_0) + \frac{1}{\pi}\sum^{+\infty}_{m =1}\left|C_{m}(r,r_0)\right| 
\nonumber\\
&<&  
\frac{1}{2 \pi} p_{R|R_0}(r|r_0)  +  \frac{1}{\pi}\frac{\sqrt{2} r}{\sigma^2 \const{L}} e^{-\frac{(r^2 + r_0^2)}{\sigma^2 \mathcal{L}}}I_0\left(\frac{2 r_0r}{\sigma^2 \mathcal{L}}\right)\sum^{+\infty}_{m =1} \frac{ \beta_{m} }{\sinh(\beta_{m})}\nonumber\\
\\
&=& k_{u}p_{R|R_0}(r|r_0),
%\label{condpdf}
\end{IEEEeqnarray*}
where $k_u$ is defined in \eqref{eq:ku}. It can be verified that $k_u<\infty$.

The upper bound on $p_{R|R_0}(r|r_0)$ follows from applying the inequality in Lemma~\ref{lem:bessel}-2 to the
conditional pdf \eqref{condR}.

\vspace{1mm}
\emph{Lower Bound.}

Applying the lower bound in \eqref{eq:up-low-cm} to the conditional pdf \eqref{condpdf} and using the 
first upper bound on $C_m(r,r_0)$ in \eqref{eq:cm-ub-1}
\begin{align}
p(r,\phi | r_0,\phi_0) &\geq \frac{1}{2 \pi} p_{R|R_0}(r|r_0) - \frac{1}{\pi}\sum^{+\infty}_{m =1}\left|C_{m}(r,r_0)\right| \nn\\
&\geq
\frac{1}{2 \pi} p_{R|R_0}(r|r_0)  - \frac{1}{\pi}\frac{\sqrt{2} r}{\sigma^2 \const{L}} e^{-\Re(a_1)(r^2 + r_0^2)}  I_0\left(\frac{2 r_0r}{\sigma^2 \mathcal{L}}\right)\sum^{+\infty}_{m =1} \frac{ \beta_{m} }{\sinh(\beta_{m})} \label{newnd}\\
&\overset{(a)}{\geq} \frac{p_{R|R_0}(r|r_0)}{2 \pi}  \hspace{-0.03cm}\left( \hspace{-0.03cm}1 \hspace{-0.03cm}- \hspace{-0.03cm}\sqrt{2}e^{-\left(\Re(a_1)-\frac{1}{\sigma^2 \mathcal{L}}\right) r_0^2} \hspace{-0.03cm}\sum^{+\infty}_{m =1}  \hspace{-0.03cm}\frac{\beta_{m} }{\sinh(\beta_{m})} \hspace{-0.03cm}\right)
\nonumber
%\label{condpdflb}
\\
&= \frac{1}{2 \pi} p_{R|R_0}(r|r_0)\left(1-\xi(r_0)\right), 
%\label{condpdflbf}
\nn
\end{align}
where $\xi(r_0)$ is defined in \eqref{eq:zeta-r0}. Step $(a)$ follows from \eqref{condR} and 
\begin{IEEEeqnarray*}{rCl}
e^{-\left(\Re(a_1)-\frac{1}{\sigma^2 \mathcal{L}}\right) (r^2+r_0^2)}\leq e^{-\left(\Re(a_1)-\frac{1}{\sigma^2 \mathcal{L}}\right) r_0^2},
\end{IEEEeqnarray*}
which holds because, from \eqref{lwbnda1}--\eqref{lwbda}, $\Re(a_1) > \frac{1}{\sigma^2\const{L}}$.
Finally, since $\Re(a_1) > \frac{1}{\sigma^2\const{L}}$, $\xi(r_0)\rightarrow 0$ as $r_0\rightarrow +\infty$.

\end{proof}

\begin{lemma}[Bounds on the Output pdf]
\label{lem:outbds}

Consider the conditional pdf \eqref{condpdf}.  Let $F_0(r_0,\phi_0)\defeq F_X(r_0, \phi_0)$ be an input cdf
and denote by $p(r,\phi;F_0)$ the corresponding output pdf:
\begin{equation*}
   p(r,\phi; F_0) = \int p(r,\phi|r_0,\phi_0) dF_0(r_0,\phi_0).
\end{equation*}
\begin{enumerate}
    \item If $F_{R_0}(r_0) \in \mathcal{P}$, then 
    \begin{eqnarray*}
    p(r,\phi; F_0) \leq \frac{2k_{u}r}{\sigma^2\mathcal{L}} e^{-\frac{r^2-2 r \rho}{\sigma^2 \mathcal{L}}},
    \end{eqnarray*}
    where recall that $\rho$ is the peak amplitude defined in \eqref{acos}.
   \item If $F_{R_0}(r_0) \in \mathcal{A}$, then
   \begin{eqnarray}
   p(r, \phi; F_0) \leq \frac{2k_{u}r}{\sigma^2\mathcal{L}} \left(e^{-\frac{r^2}{4 \sigma^2 \mathcal{L}}} + \frac{A}{\mathcal{C}\left(\frac{r}{2}\right)}\right),
   \label{coscon}
   \end{eqnarray}
where $k_u$ is defined in \eqref{eq:ku}, and recall that $A$ is the average cost defined in \eqref{ccos}.

\item If $F_{R_0}(r_0)\in\mathcal{F}$, then for large values of $r$
\begin{IEEEeqnarray*}{rCl}
p(r,\phi;F_0) \geq \frac{k_1 r}{\pi \sigma^2\mathcal{L}} e^{-\frac{r^2}{\sigma^2 \mathcal{L}}}\left(1- \xi(r)\right),   
\end{IEEEeqnarray*}
where $\xi(\cdot)$ is defined in Lemma~\ref{lembds} and $k_1 = \int_{0}^{+\infty} e^{-\frac{r_0^2}{\sigma^2 \mathcal{L}}}\,dF_{R_0}(r_0)$. Recall that $\mathcal F$ is any of the sets $\mathcal P$, $\mathcal A$ or 
    $\mathcal P\cap\mathcal A'$ defined in Section \ref{sec:main}. Furthermore, for $r \geq 0$

\begin{IEEEeqnarray*}{rCl}
p(r;F_{R_0}) \geq \frac{2 k_1 r}{\sigma^2\mathcal{L}} e^{-\frac{r^2}{\sigma^2 \mathcal{L}}},
\end{IEEEeqnarray*}
where $p(r;F_{R_0}) =  \int p_{R|R_{0}}(r|r_0)dF_{R_0}(r_0)$ and 
\end{enumerate}
\end{lemma}

\begin{proof}
We have
\begin{IEEEeqnarray*}{rCl}
     p(r,\phi;F_0) &=& \int p(r,\phi|r_0,\phi_0) dF_0(r_0,\phi_0) \nn\\
     &\leq& k_{u} \int p_{R|R_{0}}(r|r_0)dF_{R_0}(r_0) \nn\\
     &=& k_{u}\,p(r;F_{R_0}),
     \label{prelimres}
\end{IEEEeqnarray*}
where we used the upper bound \eqref{eq:cond-pdf-ub} in
Lemma~\ref{lembds}. We bound $p(r;F_{R_0})$ for the three cases  below.

\emph{Case 1) $F_{R_0}(r_0) \in \mathcal{P}$}. We have
\begin{eqnarray}
p(r; F_{R_0}) & = &  \int p_{R|R_{0}}(r|r_0)dF_{R_0}(r_0)\nn\\
&=& \frac{2r}{\sigma^2\mathcal{L}} \int^{\rho}_{0}e^{-\frac{r^2+r_0^2}{\sigma^2 \mathcal{L}}} I_0\left(\frac{2 r
r_0}{\sigma^2 \mathcal{L}}\right)\,dF_{R_0}(r_0) \nn\\
&\leq&  \frac{2r}{\sigma^2\mathcal{L}} e^{-\frac{r^2-2 r \rho}{\sigma^2 \mathcal{L}}},\label{upperppre}
\end{eqnarray}
where $\rho$ is the peak amplitude in \eqref{acos} and we used the inequality in Lemma~\ref{lem:bessel}-2, as well as
the inequality $(r-r_0)^2\geq r^2-2r\rho$ when $r_0\in [0,\rho]$.

\emph{Case 2) $F_{R_0}(r_0) \in \mathcal{A}$}. 
The average cost constraint upper bounds the tail of the input distribution as follows. For any $a\geq 0$ 
\begin{IEEEeqnarray*}{rCl}
A&\geq& \int_{0}^{+\infty}\mathcal{C}(r_0)\,dF_{R_0}(r_0)
\\
&\geq &
 \int_{a}^{+\infty}\mathcal{C}(r_0)\,dF_{R_0}(r_0)
\\
&\geq &
\mathcal{C}(a)
\int_{a}^{+\infty}dF_{R_0}(r_0),
\end{IEEEeqnarray*}
where we used the properties C1-C2 of the cost function $\mathcal{C}(r_0)$ in Section~\ref{sec:main}.
Therefore
\begin{IEEEeqnarray}{rCl}
\int_{a}^{+\infty}dF_{R_0}(r_0)
&\leq& 
\frac{A}
{\mathcal{C}(a)}.
\label{eq:pdf-r-ub}
\end{IEEEeqnarray}
Note that, since $\mathcal C(a)=\omega(a^2)$, $\mathcal C(a)$ grows faster than $a^2$ as $a\rightarrow\infty$.

We have: 
\begin{eqnarray*}
p(r;F_{R_0}) &=& \frac{2r}{\sigma^2\mathcal{L}} \int^{+\infty}_{0}e^{-\frac{r^2+r_0^2}{\sigma^2 \mathcal{L}}}
I_0\left(\frac{2 r r_0}{\sigma^2 \mathcal{L}}\right)\,dF_{R_0}(r_0) \nn\\
&\overset{(a)}{\leq}&  \frac{2r}{\sigma^2\mathcal{L}} \int^{+\infty}_{0}e^{-\frac{(r-r_0)^2}{\sigma^2 \mathcal{L}}}
\,dF_{R_0}(r_0)\nn\\
&=& \frac{2r}{\sigma^2\mathcal{L}} \left(\int^{\frac{r}{2}}_{0}e^{-\frac{(r-r_0)^2}{\sigma^2 \mathcal{L}}}
\,dF_{R_0}(r_0) +  \int^{+\infty}_{\frac{r}{2}}e^{-\frac{(r-r_0)^2}{\sigma^2 \mathcal{L}}} \,dF_{R_0}(r_0)\right)\nn\\
&\overset{(b)}{\leq}& \frac{2r}{\sigma^2\mathcal{L}} \left(e^{-\frac{r^2}{4 \sigma^2 \mathcal{L}}} +
\int^{+\infty}_{\frac{r}{2}}\,dF_{R_0}(r_0)\right)\nn\\
&\overset{(c)}{\leq}& \frac{2r}{\sigma^2\mathcal{L}} \left(e^{-\frac{r^2}{4 \sigma^2 \mathcal{L}}} + \frac{A}{\mathcal{C}\left(\frac{r}{2}\right)}\right). 
\nn
%\label{coscon}
\end{eqnarray*}
Step $(a)$ follows from Lemma~\ref{lem:bessel}-2. 
Step $(b)$ follows because 
$\exp\left(-\frac{(r-r_0)^2}{\sigma^2 \mathcal{L}}\right)$ increases with $r_0\in [0,\frac{r}{2}]$, and 
$\exp\left(-\frac{(r-r_0)^2}{\sigma^2 \mathcal{L}}\right)\leq 1$. 
Step $(c)$ follows from \eqref{eq:pdf-r-ub} with $a=r/2$.

\emph{Case 3)} We have for large values of $r$
\begin{IEEEeqnarray*}{rCl}
 p(r,\phi;F_0) &=& \int p(r,\phi|r_0,\phi_0) dF_0(r_0,\phi_0)\\
                      &\overset{(a)}{\geq}& \int \left(\frac{1}{2 \pi} p_{R|R_0}(r|r_0)  - \frac{1}{\pi}\frac{\sqrt{2} r}{\sigma^2 \const{L}} e^{-\Re(a_1)(r^2 + r_0^2)}  I_0\left(\frac{2 r_0r}{\sigma^2 \mathcal{L}}\right)\sum^{+\infty}_{m =1} \frac{ \beta_{m} }{\sinh(\beta_{m})}\right) dF_{R_0}(r_0)\\
                      &\overset{(b)}{=}& \frac{r}{\pi \sigma^2 \mathcal{L}} e^{-\frac{r^2}{\sigma^2 \mathcal{L}}}  \int I_0\left(\frac{2 r_0r}{\sigma^2 \mathcal{L}}\right) e^{-\frac{r_0^2}{\sigma^2 \mathcal{L}}} \left(1  -  \sqrt{2}\,e^{-(\Re(a_1)-\frac{1}{\sigma^2 \const{L}})(r^2 + r_0^2)}\sum^{+\infty}_{m =1} \frac{ \beta_{m} }{\sinh(\beta_{m})} \right) dF_{R_0}(r_0)\\
                      &\overset{(c)}{\geq}& \frac{r}{\pi \sigma^2 \mathcal{L}} e^{-\frac{r^2}{\sigma^2 \mathcal{L}}}  \left(1  -  \sqrt{2}\,e^{-(\Re(a_1)-\frac{1}{\sigma^2 \const{L}})r^2}\sum^{+\infty}_{m =1} \frac{ \beta_{m} }{\sinh(\beta_{m})} \right)  \int I_0\left(\frac{2 r_0r}{\sigma^2 \mathcal{L}}\right) e^{-\frac{r_0^2}{\sigma^2 \mathcal{L}}}  dF_{R_0}(r_0)\\
                     &\overset{(d)}{\geq}& \frac{k_1 r}{\pi \sigma^2 \mathcal{L}} e^{-\frac{r^2}{\sigma^2 \mathcal{L}}}  \left(1  -  \xi(r)\right),
\end{IEEEeqnarray*}
where $\xi(\cdot)$ is defined in Lemma~\ref{lembds} and $k_1$ is defined in the statement of the lemma. Step $(a)$ follows from \eqref{newnd} and step $(b)$ follows from \eqref{condR}. Step $(c)$ holds true since $e^{-(\Re(a_1)-\frac{1}{\sigma^2 \const{L}})r_0^2} < 1$ which is true by virtue of \eqref{lwbda}.
Finally, step $(d)$ is justified by virtue of Lemma~\ref{lem:bessel}-1 and the fact that $\xi(r) \rightarrow 0$ as $r \rightarrow +\infty$.
Also we have for $r \geq 0$
\begin{IEEEeqnarray*}{rCl}
p(r;F_{R_0}) &=& \int^{+\infty}_{0}p_{R|R_0}(r|r_0) dF_{R_0}(r_0) \\
&\overset{(a)}{=}& 
\frac{2r}{\sigma^2\mathcal{L}}e^{-\frac{r^2}{\sigma^2
    \mathcal{L}}} 
\int^{+\infty}_{0}
e^{-\frac{r_0^2}{\sigma^2
    \mathcal{L}}} I_0\left(\frac{2 r r_0}{\sigma^2 \mathcal{L}}\right)
\,dF_{R_0}(r_0)\\
&\overset{(b)}{\geq}& \frac{2 k_1 r}{\sigma^2\mathcal{L}} e^{-\frac{r^2}{\sigma^2 \mathcal{L}}},   
\end{IEEEeqnarray*}
Step $(a)$ follows from \eqref{condR}. Step $(b)$
follows from Lemma~\ref{lem:bessel}-1.

\end{proof}

Define the \emph{conditional entropy density}
\begin{equation}
i(r,\phi,r_0,\phi_0)  \defeq 
\begin{cases}
- p(r,\phi|r_0,\phi_0)\ln p(r,\phi|r_0,\phi_0), & 
\textnormal{if} \quad p(r,\phi|r_0,\phi_0)>0,
\\[2pt]
0, &  \textnormal{if} \quad p(r,\phi|r_0,\phi_0)=0.
\end{cases}
\label{eq:entropy-density}
\end{equation}
%where $p(r,\phi|r_0,\phi_0)  \defeq p(r,\phi|r_0,\phi_0)$ is the conditional
%pdf \eqref{condpdf}. 
Let $F_0(r_0,\phi_0)\defeq F_X(r_0, \phi_0)$ be an input cdf and denote by $p(r,\phi;F_0)$ 
the corresponding output pdf, \ie, 
\begin{equation*}
    p(r,\phi;F_0) = \int p(r,\phi|r_0,\phi_0) dF_0(r_0,\phi_0).
\end{equation*}
We prove that $i(r,\phi, r_0, \phi_0)$ is continuous, bounded and its average
 is upper bounded by an integrable function.

\begin{lemma}
\label{lemm:entropy-density}
We have:
\begin{enumerate}
    
\item $i(r,\phi, r_0,\phi_0)$ is continuous and bounded in $(r, \phi, r_0,\phi_0)$, for all 
$r,r_0\geq 0$ and $\phi,\phi_0\in[0,2\pi)$.

\item  For any $F_{R_0}(r_0)\in \mathcal F$, 
\begin{equation*}
\left|\int i(r,\phi,r_0,\phi_0)  \,dF_0(r_0,\phi_0)\right| \leq d(r,\phi),
\end{equation*}
where $d(r,\phi)$ is independent of $F_0(r_0,\phi_0)$ and $\int_{r,\phi}d(r,\phi) dr\,d\phi < \infty$.

\end{enumerate}
\end{lemma}

\begin{proof}
The continuity of $i(r,\phi, r_0,\phi_0)$ follows from the definition of the conditional pdf \eqref{condpdf}. The boundedness is due to upper and lower bounds on $p(r,\phi|r_0,\phi_0)$ established in Lemma~\ref{lembds}. This completes Part 1) of the lemma. 

For Part 2), since $p\ln(p)<0$ when $p\in (0,1)$, we break down the integral into two parts:
\begin{equation*}
\left|\int i(r,\phi,r_0,\phi_0)  \,dF_0(r_0,\phi_0)\right| \leq  I^{+}(r,\phi) + I^{-}(r,\phi), 
\label{absineq}
\end{equation*} 
where
\begin{equation*}
 \begin{array}{ll}
        \displaystyle I^{+}(r,\phi) \defeq \int_{p(r,\phi|r_0,\phi_0) \leq 1} i(r,\phi,r_0,\phi_0) \,dF_0(r_0,\phi_0),\\
      \displaystyle   I^{-}(r,\phi) \defeq - \int_{p(r,\phi|r_0,\phi_0) > 1}i(r,\phi,r_0,\phi_0)  \,dF_0(r_0,\phi_0).    
         \end{array} 
\end{equation*}

We upper bound $I^{\pm} (r,\phi)$ for the peak amplitude and average cost constraints separately.
The upper bound on $I^{-} (r,\phi)$ is based on the upper bound on the conditional pdf in Lemma~\ref{lembds}-1. 
The upper bound on $I^{+} (r,\phi)$ is based on the lower bound on the conditional pdf in Lemma~\ref{lembds}-2.

\subsubsection*{Case 1) Peak amplitude constraint.}
Suppose that $F_0(r_0,\phi_0)$ is subjected to a peak amplitude constraint, \ie, $F_{R_0}(r_0)\in\mathcal{P}$. 
From the lower bound \eqref{condpdflbf} in Lemma~\ref{lembds}, there exists a finite $c > 0$ for which $\xi(c) < \frac{1}{2}$ such that 
\begin{IEEEeqnarray}{rCl}
\ln \frac{1}{p(r,\phi|c,\phi_0)} &\leq & 
\ln \frac{1}{\frac{1}{2 \pi} p_{R|R_0}(r|c)\left(1-\xi(c)\right)} 
\nn
\\
&\leq& \ln \frac{4 \pi}{p_{R|R_0}(r|c)}. 
\label{eqbd1}
\end{IEEEeqnarray}
On the other hand, choosing $c$ large enough and finite, we obtain
\begin{equation}
\ln \hspace{-0.03cm} \frac{1}{p(r,\phi|r_0,\phi_0)} \hspace{-0.05cm} \leq \hspace{-0.05cm}\ln \hspace{-0.03cm} \frac{1}{p(r,\phi|c,\phi_0)}, \label{eqbd2}
    \end{equation}
for $0 \leq r_0 \leq \rho$. This holds true since $p(r,\phi|c,\phi_0) \rightarrow 0$ as $c \rightarrow +\infty$. 

Inequalities~(\ref{eqbd1}) and~(\ref{eqbd2}) imply
\begin{align}
I^{+}(r,\phi) &\leq  \int_{0 < p(r,\phi|r_0,\phi_0) \leq 1} p(r,\phi|r_0,\phi_0)  \ln \frac{4 \pi}{p_{R|R_0}(r|c)}
\,dF_0(r_0,\phi_0) \nonumber\\
&\leq \left(\ln(4 \pi) -  \ln p_{R|R_0}(r|c)\right) p(r,\phi;F_0) \nonumber\\
&\leq \left(\ln(4 \pi) -  \ln \frac{2r}{\sigma^2 \mathcal{L}} +\frac{r^2 + c^2}{\sigma^2 \mathcal{L}}\right)
p(r,\phi;F_0) \label{lbconr},
\end{align}
where we used the bound $p_{R|R_0}(r|c) \geq \frac{2r}{\sigma^2 \mathcal{L}}e^{-\frac{r^2 + c^2}{\sigma^2 \mathcal{L}}}$.

Similarly, for $I^{-}(r,\phi)$ we have
 \begin{align}
 I^{-}(r,\phi) &\leq \int_{p(r,\phi|r_0,\phi_0) > 1} p(r,\phi|r_0,\phi_0)\ln k_{u} p_{R|R_0}(r|r_0) \,dF_0(r_0,\phi_0) \label{justubp}\\
 &\leq \left(\ln k_{u} + \ln \frac{2r}{\sigma^2\mathcal{L}}\right)p(r,\phi;F_0) \label{ubconr},
 \end{align}
 where we used Lemma~\ref{lembds}-1 in (\ref{justubp}) and (\ref{ubconr}) follows
 from \eqref{eq:p(r|r0)-up2}. 
 
Finally, adding \eqref{lbconr} and \eqref{ubconr}
 \begin{IEEEeqnarray}{rCl}
     I^{+}(r,\phi) + I^{-}(r,\phi) &\leq& \left(\ln(4 \pi) -  \ln \frac{2r}{\sigma^2 \mathcal{L}} +\frac{r^2 +
     c^2}{\sigma^2 \mathcal{L}}\right) p(r,\phi;F_0) + \left(\ln k_{u} + \ln
     \frac{2r}{\sigma^2\mathcal{L}}\right)p(r,\phi;F_0)\nn\\
     &\leq&  \left(\ln(4\pi k_u) +\frac{r^2 + c^2}{\sigma^2 \mathcal{L}} \right)\frac{2k_{u}r}{\sigma^2\mathcal{L}} e^{-\frac{r^2-2 r \rho}{\sigma^2 \mathcal{L}}}  \label{uniup}
     \\
     &\defeq& d(r,\phi),\nn
 \end{IEEEeqnarray}
  where we used Lemma~\ref{lem:outbds}-1 to obtain \eqref{uniup}. Clearly, $d(r,\phi)$ is 
 is an integrable function of $(r,\phi)$.

 \subsubsection*{Case 2) Average cost constraint}
 Suppose that $F_0(r_0,\phi_0)$ is subject to an average cost constraint, \ie, $F_{R_0}(r_0)\in\mathcal A$. 
 The inequality~(\ref{ubconr}) holds true and
 \begin{align}
 I^{-}(r,\phi) &\leq \left(\ln k_{u} + \ln \frac{2r}{\sigma^2\mathcal{L}}\right)p(r,\phi;F_0) \nn\\
 &\leq  \left(\ln k_{u} + \ln \frac{2r}{\sigma^2\mathcal{L}}\right) \frac{2k_{u}r}{\sigma^2\mathcal{L}} \left(e^{-\frac{r^2}{4 \sigma^2 \mathcal{L}}} + \frac{A}{\mathcal{C}\left(\frac{r}{2}\right)}\right), \label{uppmin}
 \end{align}
 where we used Lemma~\ref{lem:outbds}-2 in obtaining the last inequality. The right hand side of 
 \eqref{uppmin} is integrable because $\mathcal{C}(r)=\omega(r^2)$. 
 
 Finally, we obtain an upper bound on $I^{+}(r,\phi)$. We apply the inequality
 \begin{IEEEeqnarray}{rCl}
 |x\ln x| \leq \frac{1}{1-\delta}x^{\delta}, \quad 
 \label{eq:xln(x)}
 \end{IEEEeqnarray}
 which is valid for all $0 < x \leq 1$ and  $0 < \delta < 1$. We obtain 
\begin{align}
I^{+}(r,\phi) &=  - \int_{0 < p(r,\phi|r_0,\phi_0) \leq 1}  p(r,\phi|r_0,\phi_0)\ln p(r,\phi|r_0,\phi_0) \,dF(r_0,\phi_0) \nn\\
&\leq \frac{1}{1-\delta} \int_{0 < p(r,\phi|r_0,\phi_0) \leq 1} (p(r,\phi|r_0,\phi_0))^{\delta}\,dF(r_0,\phi_0) \nn\\
&\leq \frac{1}{1-\delta} k^{\delta}_u\int (p_{R|R_0}(r|r_0))^{\delta} \,dF(r_0) \label{justbd}\\
&\leq \frac{1}{1-\delta} k^{\delta}_u \left(\frac{2r}{\sigma^2 \mathcal{L}}\right)^{\delta}\int e^{-\delta\frac{(r-r_0)^2}{\sigma^2 \mathcal{L}}}\,dF(r_0)\label{bdI}\\
&\leq \frac{1}{1-\delta} \left(\frac{2k_ur}{\sigma^2 \mathcal{L}}\right)^{\delta}\left(e^{-\frac{\delta r^2}{4 \sigma^2 \mathcal{L}}} + \frac{A}{\mathcal{C}\left(\frac{r}{2}\right)}\right). \label{finbd}
\end{align}
The inequality \eqref{justbd} follows from the upper bound in Lemma~\ref{lembds}. Inequality~\eqref{bdI} follows from \eqref{eq:p(r|r0)-up1}.
The last inequality can be justified using the same set of inequalities that yielded inequality~(\ref{coscon}). Clearly, the right hand side in \eqref{finbd} is integrable because $\mathcal{C}(r)=\omega(r^2)$. 

The inequalities \eqref{uppmin} and \eqref{finbd} prove the desired inequality in the Part 2) of the lemma.

\end{proof}

\begin{lemma}[Continuity of the Conditional Entropy]
Let $F_{X}(r_0,\phi_0)$ be an input cdf such that $F_{R_0}(r_0) \in \mathcal{F}$. The conditional entropy $h(Y | X)$ in the channel \eqref{condpdf} is continuous function of $F_{X}$.
\label{thm:cont-cond-entr}
\end{lemma}

\begin{proof}

Let $\{X_m = (R_m, \Phi_m)\}_{m \geq 1}$ be a sequence of input variables whose cdfs are denoted respectively by $\{F_m(r_0,\phi_0)\}_{m \geq 1}$ and such that $F_{R_m}(r_0) \in \mathcal{F}$. Suppose that $\{F_m(r_0,\phi_0)\}_{m \geq 1}$ converges weakly to $F_{X}(r_0,\phi_0)\defeq F_0(r_0,\phi_0)$. We have
\begin{align}
h(Y | X) &= \int h(Y|r_0,\phi_0) \,dF_{0}(r_0,\phi_0) \nonumber\\
&= - \int \int_{r,\phi} p(r,\phi|r_0,\phi_0)\ln p(r,\phi|r_0,\phi_0) \,dr \,d\phi\,dF_{0}(r_0,\phi_0) \nonumber\\
&= \int_{r,\phi}  \int i(r,\phi,r_0,\phi_0) \,dF_{0}(r_0,\phi_0)\,dr \,d\phi \label{interap1}\\
%\label{wcvappp}\\
&= \int_{r,\phi} \lim_{m \rightarrow +\infty} \int i(r,\phi,r_0,\phi_0) \,dF_{m}(r_0,\phi_0)\,dr \,d\phi \label{wcvapp}\\
&= \lim_{m \rightarrow +\infty} \int_{r,\phi} \int i(r,\phi,r_0,\phi_0)  \,dF_{m}(r_0,\phi_0)\,dr \,d\phi \label{interap2}\\
&= \lim_{m \rightarrow +\infty} \int  \int_{r,\phi} i(r,\phi,r_0,\phi_0)  \,dr \,d\phi \,dF_{m}(r_0,\phi_0)\label{interap3}\\
&=  \lim_{m \rightarrow +\infty} h(Y | X_m). \nonumber 
\end{align}
The order of the integrals in~(\ref{interap1}) and~(\ref{interap3}) can be exchanged by applying the Fubini's theorem and using Lemma~\ref{lemm:entropy-density}-2. 
Equation~\eqref{wcvapp} is due to the weak convergence and to the fact that $i(r,\phi,r_0,\phi_0)$ is continuous and bounded using Lemma~\ref{lemm:entropy-density}-1.
The order of the limit and the integral in \eqref{wcvapp}--\eqref{interap2} can be exchanged  
by applying the dominated convergence theorem  and using Lemma~\ref{lemm:entropy-density}-2.

\end{proof}

\begin{lemma}[Continuity of the Output Entropy]
The output entropy $h(Y)$ in the channel \eqref{condpdf} is continuous function of the input cdf $F_{X}(r_0,\phi_0)$, where $F_{R_0}(r_0) \in \mathcal{F}$. 
\label{thm:cont-out-entr}
\end{lemma}

\begin{proof}

Let $F_{X}(r_0,\phi_0) \defeq F_0(r_0,\phi_0)$ be such that $F_{R_0}(r_0) \in \mathcal{F}$ and let $p(r,\phi; F_0)=\int p(r,\phi|r_0, \phi_0)dF_0(r_0,\phi_0)$ be the corresponding output pdf in the channel \eqref{condpdf}. 
%First, note that $p(r,\phi; F_0)=\int p(r,\phi|r_0, \phi_0)dF_0(r_0,\phi_0)$. 
From \eqref{condpdf}, $p(r,\phi|r_0, \phi_0)$ is a continuous bounded function of $(r_0,\phi_0)$. Thus, 
from the definition of the weak convergence, $p(r,\phi; F_0)$ is continuous in $F_0$. Hence  $p(r,\phi; F_0)\bigl|\ln(p(r,\phi; F_0)\bigr|$ is continuous in $F_0$.
Second, from Lemma~\ref{lem:outbds}-3, 
$p(r,\phi; F_0)\bigl|\ln(p(r,\phi; F_0)\bigr|$
is dominated by an integrable function of $(r,\phi)$, that is independent of $F_0$. 
Hence $h(Y)$ is continuous in output pdf.
Combining these two results proves the lemma.

\end{proof}

\section{Analyticity of the KKT Conditions}
\label{app:analyticity}

Consider $\text{LHS}_{\rho}(r_0)$ and $\text{LHS}_{A}(r_0)$ in the KKT conditions defined in \eqref{eq:KKT1} and \eqref{eq:KKT} respectively. 

\begin{lemma}[Analytic Extensions of $\text{LHS}_{\rho}(r_0)$ and $\text{LHS}_{A}(r_0)$]
\label{lem:analyticity}
There exists a $\delta > 0$ and a non-empty open connected region $\mathcal O _{\delta}$ in the complex plane containing the non-negative real line 
$\Reals^{+}_0 = \left\{z \in \Reals: z \geq 0\right\}$ such that $\text{LHS}_{\rho}(r_0)$ and $\text{LHS}_{A}(r_0)$ can be analytically extended from $r_0\in \Reals^{+}_0$ to $z \in \mathcal O _\delta$.
\end{lemma}

\begin{proof}

We break down the proof into five steps.

\subsubsection{Analyticity of $p(r,\phi|z,\phi_0)$}
\label{subana1}

Let 
\begin{equation}
s(z;r,\phi,\phi_0) =  \frac{1}{2 \pi}\left(p_{R|R_0}(r|z) \, + \sum^{+\infty}_{m = -\infty,  m \neq 0} C_{m}(r,z)e^{jm(\phi-\phi_0-\gamma z^2\mathcal{L})}\right),
\label{eq:sz}
\end{equation} 
be an extension of $p(r,\phi|r_0,\phi_0)$ defined by~(\ref{condpdf}) from $r_0\in \Reals_{0}^{+}$ 
to the complex plane $r_0\in\Complex$. We note that after straightforward manipulations of equations~(\ref{coeffC}),~(\ref{coeffa}), and~(\ref{coeffb}), it can be shown that $a_{-m} =  a^*_m$, $b_{-m} = b^*_m$ and $C_{-m}(r,r_0) = {C}_m^*(r,r_0)$, $m \in \Naturals$. We prove that $s(z;r,\phi,\phi_0)$ is an entire
function.

From Lemma~\ref{lem:bessel}-6, $I_m(z)$ is an entire function of $z$ for $m \in \integers$. Thus, $p_{R|R_0}(r|z)$ and each term in the sum in 
\eqref{eq:sz} are entire functions. Below, we show that $|C_m(r,z)e^{jm(\phi-\phi_0-\gamma z^2\mathcal{L})}|$ can be upper bounded in $\Complex$ for large values of $m$
by an absolutely summable sequence of $m$ that is independent of $z$. Thus the series \eqref{eq:sz} 
converges uniformly over $z\in\mathcal \Complex$. The sum of a uniformly convergent sequence of analytic functions is analytic.
Therefore, $s(z;r,\phi,\phi_0)$ is an analytic function on $\Complex$.

We have
\begin{IEEEeqnarray}{rCl}
 \left|C_{m}(r,z)e^{jm(\phi-\phi_0-\gamma z^2\mathcal{L})}\right| & =&  \left|r b_me^{-a_m(r^2+z^2)}I_{m}\left(2b_mzr\right) e^{-jm\gamma z^2\mathcal{L}}\right| 
 \nn\\
 &=& r|b_{|m|}|\left|I_{m}\left(2b_mzr\right)\right| e^{-\Re(a_m)r^2}e^{-\Re(a_m)\Re(z^2)} e^{(\Im(a_m)+m \gamma \mathcal{L}) \Im(z^2)}
 \label{besprop1},
%&\leq& r|b_m|\left|I_{m}\left(2b_mzr\right)\right|e^{-\frac{r^2}{\sigma^2 \mathcal{L}}}e^{-\left(\Re(a_m)\Re(z^2)-(\Im(a_m)+m \gamma \mathcal{L}) \Im(z^2)\right)} ,
%\label{besprop1}
%&\leq& r e^{-\frac{r^2}{\sigma^2 \mathcal{L}}}e^{-\left|\Im(z^2)\right|}|b_m|e^{2|b_m||z|r}, 
%\label{newjustbd}
%\\&\leq& r e^{-\frac{r^2}{\sigma^2 \mathcal{L}}}e^{-\left(\left|z_0\right| - \zeta\right)^2}|b_m|e^{2|b_m|(|z_0| + \zeta)r}
%\label{coeffCbd}
\end{IEEEeqnarray}
 %where we used the fact that $\Re(a_m) > \frac{1}{\sigma^2 \mathcal{L}}$ to write inequality~(\ref{besprop1}).
where we used the fact that $|b_m| = |b_{|m|}|$ since $b_{-m} = b^*_m$. Recall from the proof of Lemma~\ref{lembds} that $\Re(a_m) \defeq \frac{1}{\sigma^2 \mathcal{L}} t(\beta_m)$ where $t(x)$ is given by equation~(\ref{teq}) and $\beta_{m}\defeq \sqrt{\frac{m \gamma}{2}} \sigma \mathcal{L}$, $m > 0$. Also, the imaginary part of $a_m$ in \eqref{coeffa} is $\Im(a_m) \defeq\frac{1}{\sigma^2\mathcal{L}} \tau(\beta_m) $, 
where
\begin{eqnarray}
\tau(x) = \frac{x \left(\sinh(2x) - \sin(2x)\right)} {2\left(\sinh^2(x) + \sin^2(x)\right)}\nn.
\end{eqnarray}
Notice that $t(x) \equiv x$ and $\tau(x) \equiv x$ as $x \rightarrow +\infty$. Let $\epsilon > 0$, then there exists an $M>0$ such that whenever $|m| > M$, we have
\begin{align}
0<(1-\epsilon) \frac{\beta_{|m|}}{\sigma^2 \mathcal{L}} &< \Re(a_{m}) < (1+\epsilon) \frac{\beta_{|m|}}{\sigma^2 \mathcal{L}}\label{bdram}\\
0<(1-\epsilon) \frac{\beta_{|m|}}{\sigma^2 \mathcal{L}} &< \left|\Im(a_{m})\right| < (1+\epsilon) \frac{\beta_{|m|}}{\sigma^2 \mathcal{L}},\label{bdiam}
\end{align}

%Consider the function $D(x) = \frac{\tau(x)}{t(x)}$ for $x>0$, where $t(x)$ is given in equation~(\ref{teq}).
%It can be shown that $D(x)$ is positive, continuous and is upper bounded by $1$. Thus $0 < \frac{\Im(a_m)}{\Re(a_m)} = \frac{\tau(\beta_m)}{t(\beta_m)} < 1$. 
where we used the fact that $a_{-m} =  a^*_m$. Now, let $\mathcal B(z_0, \zeta)=\left\{z: \left|z-z_0\right| < \zeta\right\}$ be a neighborhood of $z_0 \in \Complex$. 
If $z\in \mathcal B(z_0,\zeta)$, then $|z^2| < (|z_0| + \zeta)^2$ and 
\begin{align}
-\Re(a_{m})\Re(z^2) &< (1+\epsilon) \frac{\beta_{|m|}}{\sigma^2 \mathcal{L}} (|z_0| + \zeta)^2 \label{bd1}\\
\Im(a_{m}) \Im(z^2) &< (1+\epsilon) \frac{\beta_{|m|}}{\sigma^2 \mathcal{L}} (|z_0| + \zeta)^2 \label{bd2}\\
-\Re(a_{m}) r^2 &< - \frac{r^2}{\sigma^2 \mathcal{L}},\label{bd3}
\end{align}
for $|m| > M$, where we used inequalities~(\ref{bdram}), ~(\ref{bdiam}), and~(\ref{lwbda}).
%This implies that $M<\infty$. Note that $\mathcal{E} \supset \Reals^{+}$, and that for any $z \in \mathcal{E}$ we have that $\Re(a_mz^2) > \left|\Im(z^2)\right| > 0$. 
%and inequality~(\ref{coeffCbd}) is due to the fact that $\left|z-z_0\right| < \zeta$. 
Also, for $|m|$ large enough, we have $|b_{|m|}| < 1$ as $|m| \rightarrow +\infty$ as inferred by equation~(\ref{unibdb1}) and
\begin{equation}
\left|I_{m}\left(2b_mzr\right)\right| = \left|I_{|m|}\left(2b_mzr\right)\right| \leq (1+\epsilon) \frac{(2|b_m| |z| r)^{|m|}}{|m|!}\leq (1+\epsilon) \frac{(2(|z_0| + \zeta) r)^{|m|}}{|m|!}, \label{bd4}
\end{equation}
where we used Lemma~\ref{lem:bessel}-3. Using the bounds in~(\ref{bd1}),~(\ref{bd2}),~(\ref{bd3}), and~(\ref{bd4}), we obtain from~\eqref{besprop1} for $|m| > M$
%\begin{IEEEeqnarray}{rCl}
 %\left|C_{m}(r,z)e^{jm(\phi-\phi_0-\gamma z^2\mathcal{L})}\right| &\leq& r e^{-\frac{r^2}{\sigma^2 \mathcal{L}}}e^{\Re(a_m)\, (2+\sigma^2 m \gamma \mathcal{L}^2)(|z_0| + \zeta)^2}|b_m|\left|I_{m}\left(2b_mzr\right)\right|.
%\label{coeffCbd2}
%\end{IEEEeqnarray}
%Hence, equation~(\ref{coeffCbd2}) implies
\begin{IEEEeqnarray}{rCl}
 \left|C_{m}(r,z)e^{jm(\phi-\phi_0-\gamma z^2\mathcal{L})}\right| &\leq&(1+\epsilon) \,r e^{-\frac{r^2}{\sigma^2 \mathcal{L}}} \frac{(2 (|z_0| + \zeta))^{|m|} r^{|m|} }{|m|!} e^{2(1+\epsilon)\frac{\beta_{|m|} (|z_0| + \zeta)^2}{\sigma^2 \mathcal{L}}} e^{|m|\gamma \mathcal{L}(|z_0| + \zeta)^2 }\nn\\
 &\leq&(1+\epsilon) \,r e^{-\frac{r^2}{\sigma^2 \mathcal{L}}} \frac{\left(2 \, r (|z_0| + \zeta) e^{\sqrt{2 \gamma}(1+\epsilon)\frac{  (|z_0| + \zeta)^2}{\sigma}} e^{\gamma \mathcal{L}(|z_0| + \zeta)^2}\right)^{|m|}}{|m|!},
\label{coeffCbd2}
\end{IEEEeqnarray}
where we replaced $\beta_{|m|}$ by its expression and we used the fact that $\sqrt{|m|} < |m|$ for large values of $|m|$. It can be seen that the z-independent upper bound in \eqref{coeffCbd2} is summable over $m$. Thus the infinite sum in \eqref{eq:sz} is uniformly convergent for all $z\in\Complex$ and  $s(z;r,\phi,\phi_0)$ is an entire function.
%Using the upper bound on $|b_{|m|}|$ provided by~(\ref{unibdb1}) and the well known Stirling's formula $m! \equiv \sqrt{2 \pi m} \left(\frac{m}{e}\right)^m$ as $m \rightarrow +\infty$, it can be seen that the z-independent upper bound in \eqref{coeffCbd2} is summable over $m$. Thus the infinite sum in
%\eqref{eq:sz} is uniformly convergent for all $z\in\Complex$ and  $s(z;r,\phi,\phi_0)$ is an entire function.

%The upper bound \eqref{newjustbd} suggests to define the open subset 
%\begin{IEEEeqnarray}{rCl}
%\mathcal{E} = \left\{z \in \Complex: \Re(z^2) - M\left|\Im(z^2)\right|> 0\right\}, 
%\end{IEEEeqnarray}
%where $M = \sup_{m \geq 0} \frac{\Im(a_m)}{\Re(a_m) }$. We show that $M<\infty$.

%and inequality~(\ref{coeffCbd}) is due to the fact that $\left|z-z_0\right| < \zeta$. 

\subsubsection{Analyticity of $s(z;r,\phi,\phi_0)\ln \left(s(z;r,\phi,\phi_0)\right)$} 
We show that there exists an open covering $\mathcal{O}$ of $\Reals^+_0$ such that the function $s(z;r,\phi,\phi_0)\ln \left(s(z;r,\phi,\phi_0)\right)$ is analytic in $z \in \mathcal{O}$. 
%and such that $\Re(a_m z^2) > 0$. 
%Write
%\begin{equation*}
%u(z) \hspace{-0.02cm}= \hspace{-0.1cm}\int_{0}^{2 \pi}\hspace{-0.1cm}\int^{2 \pi}_{0}\hspace{-0.1cm}\int_{0}^{+\infty}\hspace{-0.15cm}s(z;r,\phi,\phi_0)\ln \left(s(z;r,\phi,\phi_0)\right) dr\,d\phi \,d\phi_0.
%\end{equation*}
We proved in the previous part that $s(z;r,\phi,\phi_0)$ is analytic in $z \in \Complex$. As for the analyticity of $s(z;r,\phi,\phi_0) \ln \left(s(z;r,\phi,\phi_0)\right)$:
\begin{itemize}
[leftmargin=*, label={--}]

\item We have
\begin{IEEEeqnarray}{rCl}
  p(r,\phi | r_0, \phi_0)&=&     p(r | r_0, \phi_0) p(\phi | r, r_0, \phi_0)
  \nn\\
&=&   p(r | r_0) p(\phi | r, r_0, \phi_0),
\label{eq:condpdf-bayes}
\end{IEEEeqnarray}
where \eqref{eq:condpdf-bayes} follows because $p(r|r_0,\phi_0)=p(r|r_0)$ from \eqref{eq:zerodiso}. 
From \eqref{condpdf}, $p(r|r_0)>0$ for $r \in (0, \infty)$ and $r_0 \geq 0$. For the phase conditional pdf in \eqref{eq:condpdf-bayes}, we apply 
the Karhunen-Lo\'eve expansion in \cite[Sec. V]{yousefi2011opc}, to write the phase $\Phi$ as $\Phi_0$ plus a weighted 
sum $S$ of $M\to\infty$ independent 
non-central chi-squared random variables with two degrees-of-freedom and a centrality parameter depending on $r_0$. 
%It can be verified that the convolution of $M > 1$ non-central chi-squared pdfs is a unimodal pdf strictly 
%bounded from below, independent of $M$. Thus $p_S(s)$ is a unimodal positive pdf in $s\in(0,\infty)$. 
The PDF $p_S(s;M)$ is a generalized chi-squared random variable. It can be verified that 
$p_S(s;M)>c$ for $s\in (0,\infty)$, where $c$ is independent of $M$, so that $p_S(s;\infty)>0$.
The $p(\phi |r_0, \phi_0)$ is the wrapped distribution of $p_S(s;\infty)$, thus $p(\phi|\phi_0,r_0)>0$, $r_0\in[0,\infty)$, $\forall\phi_0 \in [0, 2\pi)$. 
Finally, using \cite[Eq. 7]{yousefi2016cap}, fixing $R$ fixes 
only one term in $S$, leaving $M-1$ independent terms. Thus $p(\phi|\phi_0,r_0, r)>0$. Summarizing, if $r\in(0,\infty)$, then
\begin{IEEEeqnarray*}{rCl}
  p(r,\phi | r_0, \phi_0)>0, \quad \forall r_0 \in[0,\infty),\quad \forall\phi_0.
\end{IEEEeqnarray*}

\item For any $z_l \in \Reals^+_0$, there exits an open ball $\mathcal{B}_l(z_l,\zeta_{l})$ such
that $\Re\left(s(z;r,\phi,\phi_0)\right) > 0$, when $z \in \mathcal{B}_l$. This is true because 
$s(r_0;r,\phi,\phi_0) = p(r,\phi|r_0,\phi_0)$ is positive in $0 \leq r_0 < \infty$ (up to a set of Lebesgue measure zero) and continuous on $\Complex$. 
%The positivity of the conditional pdf follows from the cascade model of the channel presented in~\cite[Eq.~3]{yousefi2011opc}). 

\item Therefore, there exists a sequence $\{z_i\}_{i \geq 1}$ of real non-negative numbers and an open covering  $\left\{\mathcal{O} = \cup_{i \geq 1}\mathcal{B}_i\right\} \subset \Complex$ of the non-negative real line such that $\Re\left(s(z;r,\phi,\phi_0)\right) > 0$.

\item Taking the principal branch of the logarithm, $\log(z)$ is analytic in $\Complex\backslash\Reals_0^+$. Since $\Re\left(s(z;r,\phi,\phi_0)\right) > 0$, we obtain that 
$s(z;r,\phi,\phi_0)\ln \left(s(z;r,\phi,\phi_0)\right)$ is analytic in $z \in \mathcal{O}$.
\end{itemize}

It follows that $s(z;r,\phi,\phi_0)\ln \left(s(z;r,\phi,\phi_0)\right)$ is analytic in $z \in \mathcal{O}$. 

\subsubsection{Analyticity of $\int^{2 \pi}_{0} h(R,\Phi|z,\phi_0)\,d\phi_0$} 
\label{subana2}
We now prove that 
\begin{IEEEeqnarray*}{rCl}
 u(z) &=& \int^{2 \pi}_{0} h(R,\Phi|z,\phi_0)\,d\phi_0\nn\\
 &=&\int_{0}^{2 \pi}\hspace{-0.1cm}\int^{2 \pi}_{0} \int_{0}^{+\infty}s(z;r,\phi,\phi_0)
 \ln \left(s(z;r,\phi,\phi_0)\right) dr\,d\phi \,d\phi_0 \label{princbr},
 \end{IEEEeqnarray*}
 is analytic in $z \in \mathcal{O}$. The proof is based on the Morera's theorem.

\emph{Morera's Theorem  \cite{Sil}.}
If $f(z)$ is continuous in an open region $\mathcal{O}\subseteq\Complex$ and $\int_{\gamma} f(z)dz=0$ for any closed triangular contour $\gamma$ in $\mathcal{O}$, then $f(z)$ is analytic on $\mathcal{O}$.
\qed

We first show that $u(z)$ is continuous in $z\in \mathcal{O}$. If $z_0 \in \mathcal{O}$, then
\begin{align}
\lim_{z \rightarrow z_0} u(z) &= \lim_{z \rightarrow z_0} \int_{0}^{2 \pi}\hspace{-0.1cm}\int^{2 \pi}_{0}\hspace{-0.1cm}\int_{0}^{+\infty}\hspace{-0.15cm}s(z;r,\phi,\phi_0)\ln \left(s(z;r,\phi,\phi_0)\right) dr\,d\phi \,d\phi_0 \nn\\
&= \int_{0}^{2 \pi}\hspace{-0.1cm}\int^{2 \pi}_{0}\hspace{-0.1cm}\int_{0}^{+\infty}\hspace{-0.15cm} \lim_{z \rightarrow z_0}s(z;r,\phi,\phi_0)\ln \left(s(z;r,\phi,\phi_0)\right) dr\,d\phi \,d\phi_0 \label{morcont0}\\
&= \int_{0}^{2 \pi}\hspace{-0.1cm}\int^{2 \pi}_{0}\hspace{-0.1cm}\int_{0}^{+\infty}\hspace{-0.15cm} s(z_0;r,\phi,\phi_0)\ln \left(s(z_0;r,\phi,\phi_0)\right) dr\,d\phi \,d\phi_0 \label{morcont1}
\\
&=
u(z_0),\label{morcont999}
\end{align}
where equation~(\ref{morcont1}) holds since $s(z;r,\phi,\phi_0)\ln(s(z;r,\phi,\phi_0))$ is continuous in $z$, considering the 
convention \eqref{eq:entropy-density} when $s=0$. We justify changing the order of the limit and the sum in \eqref{morcont0} by finding, for large values of $r$, an upper bound on $|s(z;r,\phi,\phi_0)\ln s(z;r,\phi,\phi_0)|$ 
that is integrable in $(r, \phi,\phi_0)$ and independent of $z$. Let $\mathcal B(z_0,\zeta) \in \mathcal{O}$ be a neighborhood of $z_0$ and consider $z \in \mathcal{B}(z_0,\zeta)$. The bound in~\eqref{coeffCbd2} implies that $\sum^{+\infty}_{m = -\infty,  m \neq 0} C_{m}(r,z)e^{jm(\phi-\phi_0-\gamma z^2\mathcal{L})} = o\left(\frac{1}{r^2}\right)$. Furthermore, using~(\ref{condR}) 
\begin{equation*}
\left|p_{R|R_0}(r|z)\right| = \frac{2r}{\sigma^2\mathcal{L}}\left|e^{-\frac{r^2+z^2}{\sigma^2 \mathcal{L}}}\right| \left|I_0\left(\frac{2 r z}{\sigma^2 \mathcal{L}}\right)\right| \leq \frac{2r}{\sigma^2\mathcal{L}}e^{-\frac{r^2}{\sigma^2 \mathcal{L}}}e^{\frac{(|z_0| + \zeta)^2}{\sigma^2 \mathcal{L}}} e^{\frac{2(|z_0| + \zeta)r}{\sigma^2\mathcal{L}}} = o\left(\frac{1}{r^2}\right),
\end{equation*}
where we used Lemma~\ref{lem:bessel}-2 to write the inequality. Hence, equation~(\ref{eq:sz}) implies
\begin{equation}
   |s(z;r,\phi,\phi_0)| = o\left(\frac{1}{r^2}\right), 
   \label{bdcom}
\end{equation}
%where the sum over $m$ is finite considering (\ref{unibdb1}).
Note that, if $x\in\Complex$, then, since $\ln(x)=\ln |x|+j\textnormal{arg}(x)$, $|\ln(x)| \equiv |\ln |x||$ as $x\to 0$.
%Notice that since $|s(z;r,\phi,\phi_0)| \rightarrow 0$ at large values of $r$, then 
Thus $\bigl|s(z;r,\phi,\phi_0) \ln (s(z;r,\phi,\phi_0)\bigr| \equiv  
\Bigl||s(z;r,\phi,\phi_0)| \ln |s(z;r,\phi,\phi_0)|\Bigr|$ as $r \rightarrow +\infty$. We find an integrable upper bound on $\Bigl||s(z;r,\phi,\phi_0)|\ln |s(z;r,\phi,\phi_0)|\Bigr|$. When $r$ is large, $|s(z;r,\phi,\phi_0)| < 1$ and we apply the inequality \eqref{eq:xln(x)} to obtain
%$|x \ln x| \leq \frac{1}{1-\delta} x^{\delta}$ for $0 < x \leq 1$ and for $0 < \delta < 1$ to get 
\begin{equation}
    \left||s(z;r,\phi,\phi_0)| \ln |s(z;r,\phi,\phi_0)|\right| \leq \frac{1}{1-\delta}|s(z;r,\phi,\phi_0)|^{\delta}, 
    \label{newbdcom}
\end{equation}
for any $\delta\in(0,1)$.  Choosing $\delta \in (1/2,1)$ and using \eqref{bdcom}, we obtain an upper bound on~\eqref{newbdcom} that is integrable in $(r,\phi,\phi_0)$ and independent of $z$. 

For the second part of Morera's theorem, we consider the integral of $u(z)$ over the boundary $\partial \Delta$ of a compact triangle $\Delta \subset \mathcal{O}$. We have  
\begin{align}
\int_{\partial \Delta} u(z) \,dz &= \int^{2 \pi}_{0}\hspace{-0.1cm} \int^{2 \pi}_{0}\hspace{-0.1cm}\int_{0}^{+\infty}\hspace{-0.15cm}\int_{\partial \Delta} s(z;r,\phi,\phi_0)\ln \left(s(z;r,\phi,\phi_0)\right) dz\,dr\,d\phi \,d\phi_0 \nn\\
&= 0 \label{morana1}.
\end{align}
Exchanging the integration order is justified from~\eqref{bdcom} and~\eqref{newbdcom} by Fubini's theorem. It is shown in part 2) that $s(z;r,\phi,\phi_0)\ln \left(s(z;r,\phi,\phi_0)\right)$ is analytic in $z \in \mathcal{O}$. Hence, its integral over a closed contour in $\mathcal{O}$ is zero which justifies~(\ref{morana1}). 

Equations~(\ref{morcont999}) and~(\ref{morana1}) imply that $u(z)$ is analytic on $\mathcal{O}$ by Morera's theorem.

\subsubsection{Analyticity of $\int_{0}^{+\infty}p_{R|R_0}(r|z) \ln\left(p_{R}(r;F_0^*)\right)\,dr$} 
\label{subana3}

The function
\begin{equation*}
\label{condRext}
p_{R|R_0}(r|z) = \frac{2r}{\sigma^2\mathcal{L}}e^{-\frac{r^2+z^2}{\sigma^2 \mathcal{L}}} I_0\left(\frac{2 r z}{\sigma^2 \mathcal{L}}\right),
\end{equation*}
is analytic in $z \in \Complex$. Consider
\begin{equation*}
w(z) = \int_{0}^{+\infty}p_{R|R_0}(r|z) \ln\left(p_{R}(r;F_0^*)\right)\,dr.
\end{equation*}
Applying the Morera's theorem and using the exponential decay of $p_{R|R_0}(r|z)$ in $r$ along with its analyticity in $z$, it can be shown that $w(z)$ is analytic on $\Complex$. The steps are similar to the proof of analyticity in part 3) above.

\subsubsection{Analyticity of $\text{LHS}_A(z)$  and $\text{LHS}_{\rho}(z)$}

Let $\text{LHS}_{A}(z)$ be an extension of $\text{LHS}_{A}(r_0)$ (defined in \eqref{eq:KKT}) to the complex plane:
\begin{equation*}
 \text{LHS}_{A}(z) =  \nu (\mathcal{C}(z) - A) + C + \int^{+\infty}_{0} p \left(r|z\right) \ln p(r;F_{0}^{*})\,dr + \frac{1}{2 \pi} \int_{0}^{2 \pi} h\left(R,\Phi|z,\phi_{0}\right) \, d\phi_0.
  \label{KKText}
\end{equation*}

From the results of the parts 3) and 4) above, and the property C2 of $\mathcal C(r_0)$ stated in Section~\ref{sec:main},
we obtain that $\text{LHS}_{A}(z)$ is analytic when $z \in \mathcal{O}_\delta \defeq  \mathcal{O} \cap \mathcal{S}_\delta$. 
In a similar manner, we have that $\text{LHS}_{\rho}(z)$ (defined in \eqref{eq:KKT1}) is analytically extendable to $\mathcal{O}$, hence to $\mathcal{O}_\delta$. 
%We denote by $\mathcal{H} =  \mathcal{O}_{0} \cap \mathcal{S}_{\delta}$ the region where both $\text{LHS}_{A}(z)$ and $\text{LHS}_{\rho}(z)$ are analytic. Clearly, $\mathcal{H}$ is connected and that $\mathcal{H} \supset \Reals^{+}$.
\end{proof}

\section{Bounds on the KKT Conditions}
\label{app1}

\begin{lemma}[Upper Bound on $\text{LHS}_{\rho}(r_0)$]
\label{lemupplhs}
Consider $\textnormal{LHS}_{\rho}(r_0)$ (defined in \eqref{eq:KKT1}) for large $r_0$. For any $0 < \epsilon < 1$, there exists a $K > 0$ such that 
\begin{align*}
\textnormal{LHS}_{\rho}(r_0) \leq C + \ln \left(\frac{1}{K}\right) + \frac{r_0^2}{\sigma^2 \mathcal{L}} - \ln\left(1-\xi(r_0)\right) + \left(\rho - (1-\epsilon)r_0\right)\sqrt{\frac{\pi}{\sigma^2 \mathcal{L}}} \text{L}_{\frac{1}{2}}\left(-\frac{r_0^2}{\sigma^2 \mathcal{L}}\right),
\end{align*}
where $\xi(r_0)$ is defined in Lemma~\ref{lembds} and $ \text{L}_{\frac{1}{2}}(\cdot)$ is a Laguerre polynomial.
\end{lemma}
\begin{proof}

Let $F_0^*\defeq F_{R^*_0}(r_0) \in \mathcal{P}$ be the optimal input cdf. Using the upper bound on $p(r,F_0^*)$ given in~(\ref{upperppre}), we have
\begin{align}
&\int^{+\infty}_{0} p_{R|R_0}\left(r|r_0\right) \ln p(r;F_{0}^{*})\,dr \nonumber\\
%&\leq \int^{+\infty}_{0} p_{R|R_0}\left(r|r_0\right) \ln q(r)\,dr \nonumber\\
&\leq \ln\left(\frac{2}{\sigma^2 \mathcal{L}}\right) + \int_{0}^{\infty}\ln(r) p_{R|R_0}(r|r_0)\,dr  -\frac{1}{\sigma^2 \mathcal{L}}\int_{0}^{\infty}r^2 p_{R|R_0}(r|r_0)\,dr +  \frac{2 \rho}{\sigma^2 \mathcal{L}} \int_{0}^{\infty}r p_{R|R_0}(r|r_0)\,dr\nonumber\\
&= \ln\left(\frac{2}{\sigma^2 \mathcal{L}}\right) + \int_{0}^{\infty}\ln(r) p_{R|R_0}(r|r_0)\,dr  -\frac{1}{\sigma^2 \mathcal{L}} \left(\sigma^{2} \mathcal{L} + r^2_0\right) + \rho \sqrt{\frac{\pi}{\sigma^2 \mathcal{L}}} \text{L}_{\frac{1}{2}}\left(-\frac{r_0^2}{\sigma^2 \mathcal{L}}\right), \label{iub}  
\end{align}
where we used the fact that $p_{R|R_0}(r|r_0)$ is a Rician pdf with parameters
$\left(r_0,\frac{\sigma^2 \mathcal{L}}{2}\right)$. The first two moments of this pdf are 
\begin{IEEEeqnarray}{rCl}
  \int\limits_0^\infty r p_{R|R_0}(r|r_0)\der r&=& \frac{\sigma \sqrt{\pi \mathcal{L}}}{2} \,\text{L}_{\frac{1}{2}}\left(\frac{- r_0^2}{\sigma^2 \mathcal{L}}\right), \label{fmric}
\\
  \int\limits_0^\infty r^2 p_{R|R_0}(r|r_0)\der r&=&\sigma^{2}
  \mathcal{L} + r^2_0 \label{smric}.
\end{IEEEeqnarray}

Furthermore, let $r_0>0$ be a large number and let $0 < \epsilon < 1$. Using the lower bound on the
conditional pdf in Lemma~\ref{lembds}-2 we have
\begin{align}
&\frac{1}{2\pi}\int_{0}^{2 \pi} h\left(R,\Phi|r_{0},\phi_{0}\right)\,d\phi_0 \nn\\
&=   - \frac{1}{2\pi} \hspace{-0.07cm} \int_{0}^{2 \pi} \hspace{-0.12cm} \int^{+\infty}_{0} \hspace{-0.12cm}\int_{0}^{2 \pi} p(r,\phi | r_0,\phi_0) \ln p(r,\phi | r_0,\phi_0) d\phi dr d\phi_0 \nonumber\\
&\overset{(a)}{\leq} \ln(2 \pi)- \ln\left(1-\xi(r_0)\right) - \int^{+\infty}_{0}p_{R|R_0}(r|r_0)\ln\left(p_{R|R_0}(r|r_0)\right)\,dr \nn\\
&= - \ln \left(\frac{1}{\pi \sigma^2 \mathcal{L}}\right) - \int^{+\infty}_{0} \ln (r) p_{R|R_0}(r| r_0) \,dr + \frac{1}{\sigma^2 \mathcal{L}} \left(r_0^2 + \int_0^{+ \infty} r^2 p_{R|R_0}(r| r_0) \,dr\right) \nonumber\\
&\qquad \qquad \qquad \qquad \qquad \qquad  \qquad \qquad \qquad \qquad - \int^{+\infty}_{0} \ln \left(I_0\left(\frac{2rr_0}{\sigma^2 \mathcal{L}}\right)\right) p_{R|R_0}(r| r_0) \,dr - \ln\left(1-\xi(r_0)\right)\nonumber \\
&\leq \ln \left(\pi \sigma^2 \mathcal{L}\right) \hspace{-0.05cm} -  \hspace{-0.05cm}\int^{+\infty}_{0}  \hspace{-0.05cm} \ln (r) p_{R|R_0}(r| r_0) \,dr   \hspace{-0.05cm} -  \hspace{-0.05cm} \ln\left(1-\xi(r_0)\right) + \frac{ 2 r_0^2 + \sigma^2 \mathcal{L}}{\sigma^2 \mathcal{L}}- \int^{+\infty}_{0} \ln \left(Ke^{(1-\epsilon)\frac{2rr_0}{\sigma^2 \mathcal{L}}}\right) p_{R|R_0}(r| r_0) \,dr \label{lastlab}\\
&= \ln \left(\frac{\pi \sigma^2 \mathcal{L}}{K}\right) \hspace{-0.05cm} - \hspace{-0.05cm} \int^{+\infty}_{0} \hspace{-0.05cm}\ln (r) p_{R|R_0}(r| r_0) \,dr \hspace{-0.05cm} - \hspace{-0.05cm} \ln\left(1-\xi(r_0)\right) + \frac{ 2 r_0^2 + \sigma^2 \mathcal{L}}{\sigma^2 \mathcal{L}} - (1-\epsilon)\frac{2 r_0}{\sigma^2 \mathcal{L}} \int^{+\infty}_{0}r p_{R|R_0}(r| r_0) \,dr \nonumber\\
&=  \ln \left(\frac{\pi \sigma^2 \mathcal{L}}{K}\right) \hspace{-0.05cm} - \hspace{-0.05cm} \int^{+\infty}_{0} \hspace{-0.05cm} \ln (r) p_{R|R_0}(r| r_0) \,dr \hspace{-0.05cm} - \hspace{-0.05cm}  \ln\left(1-\xi(r_0)\right) + \frac{1}{\sigma^2 \mathcal{L}} \left(2 r_0^2 + \sigma^2 \mathcal{L}\right)  - (1-\epsilon)r_0 \sqrt{\frac{\pi}{\sigma^2 \mathcal{L}}} \text{L}_{\frac{1}{2}}\left(-\frac{r_0^2}{\sigma^2 \mathcal{L}}\right),\label{condentub}
\end{align}
where step (a) is due to Lemma~\ref{lembds}-2 and where we used Lemma~\ref{lem:bessel}-4 for some $K>0$, and equations~(\ref{fmric}) and~(\ref{smric}) in order to write~(\ref{condentub}). Finally, combining (\ref{iub}) and~(\ref{condentub}), we have for sufficiently large $r_0$:

\begin{align*}
 \text{LHS}_{\rho}(r_0) &=  C  - \ln(2\pi) + \int^{+\infty}_{0} p \left(r|r_0\right) \ln p(r;F_{0}^{*})\,dr + \frac{1}{2 \pi} \int_{0}^{2 \pi} h\left(R,\Phi|r_{0},\phi_{0}\right) \, d\phi_0 \nonumber\\
 &\leq C  - \ln(2\pi) + \ln\left(\frac{2}{\sigma^2 \mathcal{L}}\right) + \int_{0}^{\infty}\ln(r) p_{R|R_0}(r|r_0)\,dr -\frac{1}{\sigma^2 \mathcal{L}} \left(\sigma^{2} \mathcal{L} +r^2_0\right) \nn\\
 & \qquad + \rho \sqrt{\frac{\pi}{\sigma^2 \mathcal{L}}} \text{L}_{\frac{1}{2}}\left(-\frac{r_0^2}{\sigma^2 \mathcal{L}}\right)+\ln \left(\frac{\pi \sigma^2 \mathcal{L}}{K}\right) - \int^{+\infty}_{0} \ln (r) p_{R|R_0}(r| r_0) \,dr - \ln\left(1-\xi(r_0)\right) \nn\\
 & \qquad \qquad + \frac{1}{\sigma^2 \mathcal{L}} \left(2 r_0^2 + \sigma^2 \mathcal{L}\right) - (1-\epsilon)r_0 \sqrt{\frac{\pi}{\sigma^2 \mathcal{L}}} \text{L}_{\frac{1}{2}}\left(-\frac{r_0^2}{\sigma^2 \mathcal{L}}\right) \nonumber\\
 &= C + \ln \left(\frac{1}{K}\right) + \frac{r_0^2}{\sigma^2 \mathcal{L}} - \ln\left(1-\xi(r_0)\right) + \left(\rho - (1-\epsilon)r_0\right)\sqrt{\frac{\pi}{\sigma^2 \mathcal{L}}} \text{L}_{\frac{1}{2}}\left(-\frac{r_0^2}{\sigma^2 \mathcal{L}}\right). 
\end{align*}
\end{proof}

\begin{lemma}[Lower Bound on $\textnormal{LHS}_{A}(r_0)$]
\label{lemlowlhs}
The LHS of ~(\ref{eq:KKT}) satisfies:
\begin{align*}
\textnormal{LHS}_{A}(r_0) &> \nu (\mathcal{C}(r_0) - A) + C + \ln\left(\frac{k_1}{ 2 \pi k_u}\right)  + \frac{1}{\sigma^2 \mathcal{L}} r_0^2  - r_0 \sqrt{\frac{\pi}{\sigma^2 \mathcal{L}}} \text{L}_{\frac{1}{2}}\left(-\frac{r_0^2}{\sigma^2 \mathcal{L}}\right),
\end{align*}
where $k_u$ is defined in Lemma~\ref{lembds}, $k_1 = \int_{0}^{+\infty} e^{-\frac{r_0^2}{\sigma^2
    \mathcal{L}}}\,dF_{R_0^*}(r_0)$, and $\text{L}_{\frac{1}{2}}(\cdot)$ is a Laguerre polynomial.
\end{lemma}
\begin{proof}

Let $F_0^*\defeq F_{R^*_0}(r_0) \in \mathcal{A}$ be the optimal input cdf. The first integral in $\text{LHS}_{A}(r_0)$ is
lower bounded by
\begin{align}
\int^{+\infty}_{0} p_{R|R_0}\left(r|r_0\right) \ln p(r;F_{0}^{*})\,dr &\geq \hspace{-0.05cm}\int\limits^{\infty}_{0} \ln (r) p_{R|R_0} \left(r|r_0\right)\der r  +  \ln \left(\frac{2 k_1}{\sigma^2 \mathcal{L}}\right) 
 -  \frac{\sigma^{2} \mathcal{L} + r^2_0}{\sigma^2 \mathcal{L}},
 \label{lowbd101}
\end{align}
where we used the lower bound on $p(r,F_0^*)$ in Lemma~\ref{lem:outbds}-3 and \eqref{smric}.

The second integral in $\text{LHS}_{A}(r_0)$ is lower bounded as
\begin{align}
  \int\limits_{0}^{2 \pi} h\left(R,\Phi|r_{0},\phi_{0}\right) \, d\phi_0
&\overset{(b)}{>}  \hspace{-0.05cm} - \hspace{-0.1cm}\int\limits_{0}^{2 \pi}  \hspace{-0.1cm} \int\limits^{+\infty}_{0}  \hspace{-0.05cm} \int\limits_{0}^{2 \pi}  \hspace{-0.05cm} p(r,\phi | r_0,\phi_0) \ln \left(k_u p_{R|R_0}(r|r_0)\right)  \hspace{-0.05cm} \,d\phi \,dr \,d\phi_0 \nonumber\\
&= - \int\limits^{+\infty}_{0} \ln \left(k_u p_{R|R_0}(r|r_0)\right)
\der r \int\limits_{0}^{2 \pi}  \int\limits_{0}^{2 \pi}p(r,\phi | r_0,\phi_0) \,d\phi\,d\phi_0 \nonumber\\
&= - 2\pi \int\limits^{+\infty}_{0} \ln \left(k_u p_{R|R_0}(r|r_0)\right) p_{R|R_0}(r| r_0) \,dr \nonumber\\
&= 2\pi \left(h(R|r_0) - \ln(k_u)\right) , \label{intmbd}
\end{align}
where step $(b)$ is due to Lemma~\ref{lembds}-1. From \eqref{condR}, $h(R|r_0)$ can be lower-bounded as:
\begin{align}
h(R|r_0) &= -\int^{+\infty}_{0} p_{R|R_0}(r| r_0) \ln \left(p_{R|R_0}(r|r_0)\right) \,dr \nonumber\\
&= -\ln \left(\frac{2}{\sigma^2 \mathcal{L}}\right)  + \frac{1}{\sigma^2 \mathcal{L}} \left(r_0^2 + \int_0^{+ \infty} r^2 p_{R|R_0}(r| r_0) \,dr\right)  - \int^{+\infty}_{0} \ln \left(r I_0\left(\frac{2rr_0}{\sigma^2 \mathcal{L}}\right)\right) p_{R|R_0}(r| r_0) \,dr \nonumber \\
&\geq \ln \left(\frac{\sigma^2 \mathcal{L}}{2}\right) - \int^{+\infty}_{0} \ln (r) p_{R|R_0}(r| r_0) \,dr + \frac{1}{\sigma^2 \mathcal{L}} \left(2 r_0^2 + \sigma^2 \mathcal{L}\right) -  \frac{2r_0}{\sigma^2 \mathcal{L}} \int^{+\infty}_{0} r p_{R|R_0}(r| r_0) \,dr \label{bdbes}\\
&= \ln \left(\frac{e\, \sigma^2 \mathcal{L}}{2}\right) - \int^{+\infty}_{0} \ln (r) p_{R|R_0}(r| r_0) \,dr  + \frac{2}{\sigma^2 \mathcal{L}} r_0^2 - r_0 \sqrt{\frac{\pi}{\sigma^2 \mathcal{L}}} \text{L}_{\frac{1}{2}}\left(-\frac{r_0^2}{\sigma^2 \mathcal{L}}\right),
\label{bdbes2}
\end{align} 
where we used Lemma~\ref{lem:bessel}-2 in \eqref{bdbes} and \eqref{fmric} in \eqref{bdbes2}. Substituting \eqref{bdbes2} into 
\eqref{intmbd}
\begin{align}
\frac{1}{2 \pi} \int_{0}^{2 \pi} h\left(R,\Phi|r_{0},\phi_{0}\right) \, d\phi_0 >  \ln \left(\frac{e\,\sigma^2 \mathcal{L}}{2  k_u}\right) + \frac{2}{\sigma^2 \mathcal{L}} r_0^2 &- \int^{+\infty}_{0} \ln (r) p_{R|R_0}(r| r_0) \,dr  - r_0 \sqrt{\frac{\pi}{\sigma^2 \mathcal{L}}} \text{L}_{\frac{1}{2}}\left(-\frac{r_0^2}{\sigma^2 \mathcal{L}}\right)\label{fbd}.
\end{align}

Finally, using lower bounds \eqref{lowbd101} and \eqref{fbd}, we obtain
\begin{align*}
\text{LHS}_{A}(r_0) &> \nu (\mathcal{C}(r_0) - A) + C  - \ln(2\pi) + \ln \left(\frac{2 k_1}{\sigma^2 \mathcal{L}}\right) + \ln \left(\frac{e\,\sigma^2 \mathcal{L}}{2  k_u}\right) \nonumber\\
& \qquad \qquad + \int^{+\infty}_{0} \ln (r) p_{R|R_0} \left(r|r_0\right)\,dr - \int^{+\infty}_{0} \ln (r) p_{R|R_0}(r| r_0) \,dr \nn\\
& \qquad \qquad \qquad \qquad - \frac{1}{\sigma^2 \mathcal{L}} \left(\sigma^{2} \mathcal{L} + r^2_0\right)  + \frac{2}{\sigma^2 \mathcal{L}} r_0^2 - r_0 \sqrt{\frac{\pi}{\sigma^2 \mathcal{L}}} \text{L}_{\frac{1}{2}}\left(-\frac{r_0^2}{\sigma^2 \mathcal{L}}\right)\nonumber\\
&=  \nu (\mathcal{C}(r_0) - A) + C + \ln\left(\frac{k_1}{ 2 \pi k_u}\right)  + \frac{1}{\sigma^2 \mathcal{L}} r_0^2  - r_0 \sqrt{\frac{\pi}{\sigma^2 \mathcal{L}}} \text{L}_{\frac{1}{2}}\left(-\frac{r_0^2}{\sigma^2 \mathcal{L}}\right).
\label{finalbd}
\end{align*}
\end{proof}

 \bibliographystyle{IEEEtran}
\bibliography{refs2.bib}

\begin{thebibliography}{10}
\providecommand{\url}[1]{#1}
\csname url@rmstyle\endcsname
\providecommand{\newblock}{\relax}
\providecommand{\bibinfo}[2]{#2}
\providecommand\BIBentrySTDinterwordspacing{\spaceskip=0pt\relax}
\providecommand\BIBentryALTinterwordstretchfactor{4}
\providecommand\BIBentryALTinterwordspacing{\spaceskip=\fontdimen2\font plus
\BIBentryALTinterwordstretchfactor\fontdimen3\font minus
  \fontdimen4\font\relax}
\providecommand\BIBforeignlanguage[2]{{%
\expandafter\ifx\csname l@#1\endcsname\relax
\typeout{** WARNING: IEEEtran.bst: No hyphenation pattern has been}%
\typeout{** loaded for the language `#1'. Using the pattern for}%
\typeout{** the default language instead.}%
\else
\language=\csname l@#1\endcsname
\fi
#2}}

\bibitem{kramer2015upper}
G.~Kramer, M.~I. Yousefi, and F.~Kschischang, ``Upper bound on the capacity of
  a cascade of nonlinear and noisy channels,'' in \emph{IEEE Info.\ Theory
  Workshop}, Jerusalem, Israel, Apr. 2015, pp. 1--4.

\bibitem{yousefi2015cwit}
M.~I. Yousefi, G.~Kramer, and F.~R. Kschischang, ``Upper bound on the capacity
  of the nonlinear {S}chr\"odinger channel,'' in \emph{Canadian Workshop on
  Inf.\ Theory}, St. John's, Newfoundland, Canada, July 2015, pp. 1--5.

\bibitem{yousefi2011opc}
M.~I. Yousefi and F.~R. Kschischang, ``On the per-sample capacity of
  nondispersive optical fibers,'' \emph{IEEE Trans.\ Inf.\ Theory}, vol.~57,
  no.~11, pp. 7522--7541, Nov. 2011.

\bibitem{agrawal2007}
G.~Agrawal, \emph{Nonlinear fiber optics}.\hskip 1em plus 0.5em minus
  0.4em\relax San Francisco, CA, USA: Academic, 2007.

\bibitem{mecozzi1994llh}
A.~Mecozzi, ``Limits to long-haul coherent transmission set by the {K}err
  nonlinearity and noise of the in-line amplifiers,'' \emph{IEEE J. Lightw.\
  Technol.}, vol.~12, no.~11, pp. 1993--2000, Nov. 1994.

\bibitem{turitsyn2003ico}
K.~S. Turitsyn, S.~A. Derevyanko, I.~V. Yurkevich, and S.~K. Turitsyn,
  ``Information capacity of optical fiber channels with zero average
  dispersion,'' \emph{Phys.\ Rev.\ Lett.}, vol.~91, no.~20, p. 203901, Nov.
  2003.

\bibitem{smith1971ica}
J.~G. Smith, ``The information capacity of amplitude- and variance-constrained
  scalar {G}aussian channels,'' \emph{Inf. Contr.}, vol.~18, no.~3, pp.
  203--219, Apr. 1971.

\bibitem{Hirt88}
W.~Hirt and J.~Massey, ``Capacity of the discrete-time {G}aussian channel with
  intersymbol interference,'' \emph{IEEE Trans.\ Inf.\ Theory}, vol.~34, no.~3,
  pp. 380--388, May 1988.

\bibitem{shamai1995cap}
S.~Shamai and I.~Bar-David, ``The capacity of average and peak-power-limited
  quadrature {G}aussian channels,'' \emph{IEEE Trans.\ Inf.\ Theory}, vol.~41,
  no.~4, pp. 1060--1071, July 1995.

\bibitem{Das}
A.~Das, ``Capacity-achieving distributions for non-{G}aussian additive noise
  channels,'' in \emph{IEEE Int. Symp. Info Theory}, Sorrento, Italy, June
  2000, p. 432.

\bibitem{IA01}
I.~Abou-Faycal, M.~D. Trott, and S.~Shamai, ``The capacity of discrete-time
  memoryless {R}ayleigh-fading channels,'' \emph{IEEE Trans.\ Inf.\ Theory},
  vol.~47, no.~4, pp. 1290--1301, May 2001.

\bibitem{Aslan}
A.~Tchamkerten, ``On the discreteness of capacity-achieving distributions,''
  \emph{IEEE Trans.\ Inf.\ Theory}, vol.~50, no.~11, pp. 2773--2778, Nov. 2004.

\bibitem{chan2005cap}
T.~H. Chan, S.~Hranilovic, and F.~Kschischang, ``Capacity-achieving probability
  measure for conditionally {G}aussian channels with bounded inputs,''
  \emph{IEEE Trans.\ Inf.\ Theory}, vol.~51, no.~6, pp. 2073--2088, June 2005.

\bibitem{Gursoy2005}
M.~G. Gursoy, H.~V. Poor, and S.~Verdu, ``The noncoherent {R}ician fading
  channel--{P}art {I}: structure of the capacity-achieving input,'' \emph{IEEE
  Trans.\ Wireless Commun.}, vol.~4, no.~5, pp. 2193--2206, Sept. 2005.

\bibitem{fahsj}
J.~Fahs and I.~Abou-Faycal, ``Using {H}ermite bases in studying
  capacity-achieving distributions over {AWGN} channels,'' \emph{IEEE Trans.\
  Inf.\ Theory}, vol.~58, no.~8, pp. 5302--5322, Aug. 2012.

\bibitem{Fahs2}
J.~Fahs, N.~Ajeeb, and I.~Abou-Faycal, ``The capacity of average power
  constrained additive non-{G}aussian noise channels,'' in \emph{IEEE Int.
  Conf. Telecommun.}, Beirut, Lebanon, Apr. 2012.

\bibitem{fahs2017it}
J.~Fahs and I.~Abou-Faycal, ``On properties of the support of
  capacity-achieving distributions for additive noise channel models with input
  cost constraints,'' \emph{IEEE Trans.\ Inf.\ Theory}, vol.~58, no.~2, pp.
  1178 -- 1198, Feb. 2018.

\bibitem{kramer2017autocorrelation}
G.~Kramer, ``Autocorrelation function for dispersion-free fiber channels with
  distributed amplification,'' \emph{IEEE Trans.\ Inf.\ Theory}, pp.
  5131--5155, July 2018.

\bibitem{FAF15}
J.~Fahs and I.~Abou-Faycal, ``On the finiteness of the capacity of continuous
  channels,'' \emph{IEEE Trans.\ Commun.}, vol.~64, no.~1, pp. 166--173, Jan.
  2016.

\bibitem{rudin1964principles}
W.~Rudin \emph{et~al.}, \emph{Principles of mathematical analysis}.\hskip 1em
  plus 0.5em minus 0.4em\relax McGraw-hill New York, 1964, vol.~3.

\bibitem{Sil}
H.~Silverman, \emph{Complex Variables}.\hskip 1em plus 0.5em minus 0.4em\relax
  Boston, Massachusetts, USA: Houghton Mifflin Company, 1975.

\bibitem{abra1964}
M.~Abramowitz and I.~A. Stegun, \emph{Handbook of Mathematical Functions with
  Formulas, Graphs, and Mathematical Tables}, 10th~ed., ser. Appl. Math.\hskip
  1em plus 0.5em minus 0.4em\relax Washington, D.C., USA: National Bureau of
  Standards, 1972.

\bibitem{yousefi2016cap}
M.~I. Yousefi, ``The asymptotic capacity of the optical fiber,'' \emph{Arxiv
  preprint, \textnormal{arxiv:1610.06458}}, pp. 1--12, Nov. 2016.

\end{thebibliography}
\end{document}